\documentclass{amsart}
\pdfoutput=1
\usepackage{amsmath}
\usepackage{amsthm}
\usepackage{amsfonts}
\usepackage{amssymb}
\usepackage{amsbsy}
\usepackage{amscd}
\usepackage{eucal}
\usepackage{graphicx}
\usepackage{caption}

\newtheorem{theorem}{Theorem}[section]
\newtheorem{lemma}[theorem]{Lemma}
\newtheorem{proposition}[theorem]{Proposition}

\theoremstyle{definition}
\newtheorem{definition}[theorem]{Definition}
\newtheorem*{problem}{Problem}
\newtheorem{example}[theorem]{Example}
\newtheorem{examples}[theorem]{Examples}

\newtheorem{remark}[theorem]{Remark}
\newtheorem{remarks}[theorem]{Remarks}

\numberwithin{equation}{section}

\newcommand{\beq}{\begin{equation}}
\newcommand{\eeq}{\end{equation}}

\newcommand{\Real}{\mathop{\mathrm{Re}}}
\newcommand{\Imag}{\mathop{\mathrm{Im}}}
\newcommand{\rmd}{\mathrm{d}}
\newcommand{\rmi}{\mathrm{i}}
\newcommand{\rmb}{\mathrm{b}}
\newcommand{\rmf}{\mathrm{f}}

\newcommand{\CCm}{{\C^{\scriptscriptstyle (+)}_{-1/2}}}
\newcommand{\CCp}{{\C^{\scriptscriptstyle (+)}_{1/2}}}

\newcommand{\N}{\mathbb{N}}
\newcommand{\Z}{\mathbb{Z}}
\newcommand{\R}{\mathbb{R}}
\newcommand{\C}{\mathbb{C}}
\newcommand{\SSS}{\mathbb{S}}

\newcommand{\cB}{\mathcal{B}}
\newcommand{\cF}{\mathcal{F}}

\newcommand{\cC}{\mathcal{C}}
\newcommand{\cI}{\mathcal{I}}

\newcommand{\cG}{\mathcal{G}}

\newcommand{\cA}{\mathcal{A}}
\newcommand{\cW}{\mathcal{W}}
\newcommand{\cZ}{\mathcal{Z}}
\newcommand{\cL}{\mathcal{L}}

\newcommand{\hJ}{\widehat{J}}

\newcommand{\tg}{\widetilde{g}}
\newcommand{\tog}{\widetilde{\og}}
\newcommand{\tf}{\widetilde{f}}
\newcommand{\tJ}{\widetilde{J}}

\newcommand{\oc}{\mathfrak{c}}
\newcommand{\og}{\mathfrak{g}}
\newcommand{\om}{\mathfrak{m}}
\newcommand{\oC}{\mathfrak{C}}

\newcommand{\oJ}{\mathfrak{J}}

\newcommand{\bx}{\boldsymbol{x}}

\newcommand{\bPhi}{\boldsymbol{\Phi}}
\newcommand{\bPsi}{\boldsymbol{\Psi}}

\newcommand{\sds}{\strut\displaystyle}
\def\staccrel#1#2{\mathrel{\mathop{#1}\limits_{#2}}}
\newcommand{\ud}{\frac{1}{2}}

\begin{document}

\title[Holomorphic extensions]{Holomorphic extensions associated \\ with series expansions}

\author[E. De Micheli]{Enrico De Micheli}
\address{\sl IBF - Consiglio Nazionale delle Ricerche \\ Via De Marini, 6 - 16149 Genova, Italy \\
E-mail: enrico.demicheli@cnr.it}
\author[G. A. Viano]{Giovanni Alberto Viano}
\address{\sl Dipartimento di Fisica - Universit\`a di Genova\\
Istituto Nazionale di Fisica Nucleare - Sezione di Genova \\
Via Dodecaneso, 33 - 16146 Genova, Italy}

\subjclass[2010]{Primary 30B10, 30B40, 42A32, 81T28}
\keywords{Complex and harmonic analysis, Analytic continuation, Probability and Quantum Field Theory, KMS condition}

\date{}

\begin{abstract}
We study the holomorphic extension associated
with power series, i.e., the analytic continuation from the unit disk
to the cut-plane $\C\setminus[1,+\infty)$. Analogous results are obtained
also in the study of trigonometric series:
we establish conditions on the series coefficients which are sufficient to guarantee
the series to have a KMS analytic structure.
In the case of power series we show the connection between
the unique (Carlsonian) interpolation of the coefficients of the series and the
Laplace transform of a probability distribution.
Finally, we outline a procedure which allows us to obtain a numerical approximation of the
jump function across the cut
starting from a finite number of power series coefficients. By using the same
methodology, the thermal Green functions at real time can be numerically approximated
from the knowledge of a finite number of 
noisy Fourier coefficients in the expansion of the thermal Green functions along the 
imaginary axis of the complex time plane.
\end{abstract}

\maketitle

\section{Introduction}
\label{se:introduction}
The problem of the analytic extension associated with power series traces back to classical
results by Le Roy \cite{LeRoy}, who gives appropriate conditions on the coefficients
$g_k$ of a power series $\sum_{k=0}^\infty g_k z^k$ ($z\in\C, |z|<1$) in order to
guarantee that the function $G(z)$ to which the series converges in the open unit disk
admits a holomorphic extension from the domain $|z|<1$ up to the cut-plane
$\SSS\doteq\{z\in\C\setminus [1,+\infty)\}$. The conditions required by Le Roy's theorem are that
the coefficients $g_k$ can be regarded as the restriction to the integers of a function
$\tg(\lambda)$ ($\lambda\in\C$), holomorphic in the half-plane $\Real\lambda>-\frac{1}{2}$,
and, in addition, that there exist two constants $A$ and $N$ such that
$|\tg(\lambda)|\leqslant A(1+|\lambda|)^N$ ($\Real\lambda>-\frac{1}{2}$).
Similar results are due to Lindel\"of \cite{Lindelof} and Bieberbach \cite{Bieberbach}.
More recently, Stein and Wainger \cite{Stein} obtained very deep and general results on this topic
reconsidering the problem in the framework of the Hardy space theory. They assume the
coefficients $g_k$ to be the restriction to the integers of a function $\tg(\lambda)$
(i.e., $\tg(\lambda)|_{\lambda=k}=g_k$), holomorphic in the half-plane $\Real\lambda>-\frac{1}{2}$,
which is supposed also to belong to the Hardy space $H^2(\CCm)$ with norm
$\|\tg\|_2\doteq\sup_{\ell>-\ud}\left(\int_{-\infty}^{+\infty}|\tg(\ell+\rmi\nu)|^2\,\rmd\nu\right)^{\ud}$,
$\CCm\doteq\{\lambda\in\C\,:\,\Real\lambda>-\frac{1}{2}\}$.
Correspondingly, they consider the class of functions $G(z)$, analytic in the complex
$z$-plane ($z=x+\rmi y$; $x,y\in\R$) slit along the positive real axis from $1$ to $+\infty$.
The space of functions analytic in $\SSS$ for which 
$\sup_{y\neq 0}\left(\int_{-\infty}^{+\infty}|G(x+\rmi y)|^2\,\rmd x\right)^{\ud}<\infty$
is denoted by $H^2(\SSS)$. They prove that $G\in H^2(\SSS)$ if and only if
$g_k=\tg(\lambda)|_{\lambda=k}$ with $\tg\in H^2(\CCm)$;
moreover, the jump function across the cut $[1,+\infty)$ (i.e.,
$\lim_{y\to 0,y>0}[G(x+\rmi y)-G(x-\rmi y)]$, ($x\in [1,+\infty)$) exists in $L^2$-norm.

\vskip 0.3cm

In spite of these very important results several problems remain open, some of which are:
\begin{enumerate}
\item Given a power series of the form $\sum_{k=0}^\infty g_k z^k$ ($z\in\C$, $|z|<1$),
how can we establish if the coefficients $\{g_k\}$ are the restriction to the integers 
of a function $\tg(\lambda)$ ($\lambda\in\C$) belonging to $H^2(\CCm)$?

\item In several applications (particularly in those suggested by physical problems)
the sole property of existence of the jump function in $L^2$-norm is not sufficient,
and some smoothness condition (at least $C^0$-continuity) is often required. 
It is thus relevant to explore what conditions on the coefficients
$\{g_k\}$ are sufficient to guarantee this requirement.

\item It is very important, especially in the applications, to be able to recover the jump function
across the cut from the coefficients of the power series, which are supposed to be
known. Can we provide an algorithm capable to solve this problem?

\item In the above problem of reconstructing the discontinuity function across the cut,
only a finite number of coefficients of the power series are available in the actual numerical computation. 
Furthermore, they are necessarily affected by errors (at least, by roundoff error). 
Therefore, this problem is ill-posed in the sense of Hadamard \cite{Hadamard}, and
appropriate methods of regularization must be found. Suitable smoothness conditions turn out to be 
very useful also in this connection in order to obtain stable algorithms of reconstruction.

\item A rather natural question arises: Can appropriate conditions
on the coefficients $\{g_k\}$ be found, which are sufficient to guarantee 
the jump function across the cut to represent the density of a probability distribution? 

\item Last but not least, it would be illuminating to exhibit explicit examples of power series
which admit a holomorphic extension from the open unit disk to $\SSS$,
and, further, to reconstruct explicitly the jump function across the cut from the coefficients $\{g_k\}$.
\end{enumerate}

Another case of holomorphic extension which attracted the attention of physicists from long time
is associated with trigonometric series, and leads to the reconstruction of the Green functions
in thermal Quantum Field Theory (QFT) from imaginary time to real time. 
In particular, it is crucial to find conditions on the Fourier coefficients of the expansion 
of the thermal Green functions developed along the imaginary axis of the complex time-plane, 
which are sufficient to yield the so-called KMS (Kubo, Martin, Schwinger \cite{Haag})
analytic structure in thermal QFT, which is considered a milestone in the connection
between statistical physics and QFT. Although we are far for solving this problem in its
full generality, nevertheless we can establish appropriate conditions on the coefficients of the
trigonometric expansions which are sufficient for the series to have a KMS analytic structure.

It is always of great interest in mathematics the surprising connection that often emerges
between apparently disconnected ideas and theories. Some particularly striking instances
exist in the interaction between probability theory and analysis. One of the simplest
and deepest results in this domain is the elegant proof of Weierstrass' approximation
theorem through Bernstein's polynomials \cite{Bernstein}. Another amazing connection is the relationship
between plane Brownian motion and the theory of the harmonic functions. The interested reader
can find an excellent review on these topics written by Kahane \cite{Kahane}.
``Si parva licet componere magnis''\footnote{
Vergilius, Georgica. Recognovit O. G\"uthling. IV, 176, Teubner, Leipzig, 1904.},
one can say that any result
in this direction can deserve some interest. It is in this spirit that we study 
the connection between the jump function across the cut and the density of a probability distribution.
Strangely enough, in Kahane's review paper, in spite of the very wide spectrum of his analysis,
the role of the holomorphic extension associated with series expansions is not mentioned. 
One of the purposes of the present paper is precisely showing the interest on these topics 
in pure and applied mathematics.

The paper is organized as follows. In Section \ref{se:holo} we give preparatory lemmas,
which relate the Bernoulli trials, and, accordingly, the binomial distribution to the
Bernstein polynomials and the Hausdorff conditions. The latter will be shown to be sufficient 
to yield a unique Carlsonian interpolation of the coefficients of the series expansions being considered.
In the same section we study the main analytic properties of these Carlsonian interpolations,
which turn out to be very relevant for proving the subsequent theorems on holomorphic extensions.
Section \ref{se:kms} is devoted to the analytic continuation associated with trigonometric
series, and we shall find conditions on the Fourier coefficients which are sufficient to 
prove that the sum of these series has a KMS analytic structure.
In Section \ref{se:taylor}
the holomorphic extension associated with power series from the unit disk to the cut-plane
$\SSS$ is studied. In particular, we prove that the restriction
to non-negative integers of the Laplace transform of the jump function across the cut
yields the coefficients of the power series.
Moreover, in Subsection \ref{subse:connection-jump} the relationship between the discontinuity
function across the cut and probability theory will be addressed. 
In Section \ref{se:reconstruction} we present an algorithm which provides a regularized
numerical approximation of the
jump function across the cut starting from a finite number of coefficients of the power series.
The algorithm we present differs significantly from standard regularization methods, such as the Tikhonov procedure,
and makes no use of \emph{a-priori} global bounds on the solution.
We show that an analogous algorithm works also in the case of
trigonometric series: in particular, we can construct a regularized numerical approximation of 
the thermal Green functions at real time from those at imaginary time.
Finally, the Appendix is devoted to a rapid analysis of the KMS analytic structure for the convenience of the
readers who are not familiar with QFT.

\vskip 0.3cm

\section{Preparatory lemmas}
\label{se:holo}

\subsection{Bernoulli trials, binomial distribution, Bernstein polynomials, moment sequence, and Hausdorff conditions}
\label{subse:bernoulli}

As is well known a \textit{probability space} is a triple $(\Omega,\Sigma,P)$ of a sample space $\Omega$,
a $\sigma$-algebra $\Sigma$ of sets in it and a probability measure $P$ on $\Sigma$.
A \textit{probability measure} is a function assigning a value $P\{A\} \geqslant 0$ to each set $A\in\Sigma$
such that $P\{\Omega\}=1$, with the addition rule $P\{\cup_{n} A_n\}=\sum_{n} P\{A_n\}$ holding
for every countable collection of mutually non-overlapping sets $A_n$ in $\Sigma$.
Let $x\in[0,1]$ and $X$ be a random variable on the probability space $(\Omega,\Sigma,P)$
taking only the value $0$ and $1$ according to the Bernoulli law: $P\{X=1\}=x$ and
$P\{X=0\}=1-x$ (i.e., $x$ is the probability of ``success'' and $(1-x)$ of ``failure'').
From Bernoulli's theorem it follows that the probability of $k$ successes in $n$
Bernoullian trials is given by: $\binom{n}{k} x^k (1-x)^{n-k}$ ($k=0,1,2,\ldots,n$).
Suppose now that a random variable $X$ assumes the values $0,1,2,\ldots,n$, where $n$ is a fixed
positive integer, and suppose that the probability $P\{X=k\}$ is given by:
\beq
P\{X=k\} = \binom{n}{k} x^k (1-x)^{n-k} \qquad (k=0,1,2,\ldots,n).
\nonumber\label{1.0}
\eeq
This assignment of probabilities is permissible because the sum of all the point probabilities is
\beq
\sum_{k=0}^n P\{X=k\} = \sum_{k=0}^n \binom{n}{k} x^k (1-x)^{n-k} = 1.
\nonumber\label{1.1}
\eeq
The corresponding distribution function $F_X$ is said to be a binomial distribution with
parameters $n$ and $x$. Its values may be computed by the summation formula
\beq
F_X(t) = \sum_{0\leqslant k\leqslant t} \binom{n}{k} x^k (1-x)^{n-k}.
\nonumber\label{1.1bis}
\eeq
Next, we introduce the Bernstein polynomials $B_{n,f}(x)$ for a function $f(x)$ defined
on $[0,1]$ (see \cite{Lorentz})
\beq
B_{n,f}(x) \doteq \sum_{k=0}^n f\left(\frac{k}{n}\right)\binom{n}{k} x^k (1-x)^{n-k}.
\nonumber\label{1.2}
\eeq
Evidently, if $f(x) \equiv 1$, we have:
\beq
B_{n,1}(x) = \sum_{k=0}^n \binom{n}{k} x^k (1-x)^{n-k} = 1.
\label{1.3}
\eeq
Now, let $F$ be a probability distribution concentrated on the closed interval $[0,1]$; the
$k$th moment $\mu_k$ of $F$ is defined by
\beq
\mu_k = E\left\{X^k\right\} \doteq \int_0^1 x^k \,F\{\rmd x\},
\label{1.4}
\eeq
where $E$ denotes the expectation value and $X$ is the coordinate variable. 
Taking differences, it is seen that
\beq
\Delta\mu_k \doteq \mu_{k+1}-\mu_k = -E\left\{X^k(1-X)\right\},
\nonumber\label{1.5}
\eeq
By induction, we obtain:
\beq
(-1)^r \Delta^r\mu_k = E\left\{X^k(1-X)^r\right\},
\label{1.6}
\eeq
where
\beq
\Delta^r\mu_k = \underbrace{\Delta\cdot\Delta\cdots\Delta}_{r-\mathrm{times}}\mu_k,
\label{1.7}
\eeq
$\Delta^0$ being the identity operator. Then, in view of \eqref{1.3}, the expectation of the
Bernstein polynomial $B_{n,1}(x)$ is given by
\beq
\begin{split}
E\{B_{n,1}\}
&=\int_0^1\left[\sum_{k=0}^n \binom{n}{k} x^k (1-x)^{n-k}\right]\, F\{\rmd x\} \\
&=\sum_{k=0}^n \binom{n}{k}\int_0^1 x^k (1-x)^{n-k}\,F\{\rmd x\}=1.
\end{split}
\label{1.8}
\eeq
Next, taking into account \eqref{1.6}, we obtain the formula:
\beq
E\{B_{n,1}\}=\sum_{k=0}^n \binom{n}{k} (-1)^{n-k} \Delta^{n-k}\mu_k,
\label{1.9}
\eeq
which leads us to introduce the following definition.

\begin{definition}
We denote by $w_k^{(n)}\{\mu\}$ the terms appearing on the right hand side (r.h.s.) 
of formula \eqref{1.9}: i.e.,
\beq
w_k^{(n)}\{\mu\} \doteq \binom{n}{k} (-1)^{n-k} \Delta^{n-k}\mu_k \qquad (k=0,1,\ldots,n),
\label{1.10}
\eeq
and call them the \emph{Bernstein weights} for the sequence $\{\mu_k\}_{k=0}^\infty$.
\label{def:1}
\end{definition}

It follows from formulae \eqref{1.8}, \eqref{1.9}, and \eqref{1.10} that
\beq
\sum_{k=0}^n w_k^{(n)}\{\mu\} = 1.
\label{1.11}
\eeq
We can now state the following proposition due to Hausdorff.

\begin{proposition}
\label{pro:1}
A sequence of real numbers
$\{\mu_k\}_{k=0}^\infty$ ($\mu_0=1$) represents the moments
\eqref{1.4} of some probability distribution $F$ concentrated on $[0,1]$ if and only if
the Bernstein weights $w_k^{(n)}\{\mu\}$ satisfy the condition
\beq
w_k^{(n)}\{\mu\} \geqslant 0 \qquad (n = 0,1,2,\ldots; \, k=0,1,\ldots,n).
\label{1.12}
\eeq
\end{proposition}

\begin{proof}
The proof is given in \cite[Ch. VII, Section 3, Theorem 2]{Feller}.
\end{proof}

So far we have considered a moment sequence induced by
probability theory through formula \eqref{1.4}, and, accordingly,
we have introduced Bernstein weights and Proposition \ref{pro:1}. 
However, the concept of \emph{moment sequence} can be formulated in a more general setting.

\begin{definition}
A sequence of (real) numbers $\{\mu_k\}_{k=0}^\infty$ is a moment sequence if there exists a bounded
signed measure $\alpha$ on $\left([0,1],\cB_{[0,1]}\right)$ (where $\cB_{[0,1]}$ is the Borel field consisting
of all the ordinary Borel subsets in $[0,1]$) such that for every integer $k\geqslant 0$:
\beq
\mu_k=\int_0^1 t^k\,\rmd\alpha(t).
\nonumber\label{ZZ:1}
\eeq
\label{def:2}
\end{definition}

With a small abuse of language and notation, we continue to call \emph{Bernstein weights} the
terms $w_k^{(n)}\{\mu\}$ defined in \eqref{1.10}, even when the sequence $\{\mu_k\}_{k=0}^\infty$
is induced by a bounded signed measure. 
Of course, the definition $\Delta\mu_k\doteq\mu_{k+1}-\mu_k$ and, accordingly,
equality \eqref{1.7} and the quantity $\binom{n}{k}(-1)^{n-k}\Delta^{n-k}\mu_k$ are 
well defined, but in general equality \eqref{1.11} does not necessarily hold. \\

Next, we can state the following proposition due to Hausdorff.

\begin{proposition}
\label{pro:2}
If the Bernstein weights for the sequence $\{\mu\}_{k=0}^\infty$ satisfy the condition
\beq
\sum_{k=0}^n \left|w_k^{(n)}\{\mu\}\right| < L \qquad (n = 0,1,2,\ldots; \ L = \mathrm{positive~constant}),
\nonumber\label{1.13}
\eeq
then there exists a bounded signed measure $\alpha$ on $\left([0,1],\cB_{[0,1]}\right)$,
such that for every integer $k\geqslant 0$:
\beq
\mu_k = \int_0^1 t^k \,\rmd \alpha(t).
\label{1.14}
\eeq
\end{proposition}

\begin{proof}
The proof is given in \cite[Theorem 2b, p. 103]{Widder}.
\end{proof}

\begin{proposition}
\label{pro:3}
If there exists a constant $M>0$ such that, for a given number $p>1$:
\beq
(n+1)^{(p-1)}\sum_{k=0}^n \left|w_k^{(n)}\{\mu\}\right|^p < M \qquad (n = 0,1,2,\ldots),
\label{1.15}
\eeq
then $\alpha(t)$ in representation \eqref{1.14} is the integral of a function of class $L^p$: i.e., 
there exists $\varphi \in L^p(0,1)$ such that
\beq
\mu_k = \int_0^1 t^k \varphi(t)\,\rmd t \qquad (k=0,1,2,\ldots).
\label{1.16}
\eeq
\end{proposition}

\begin{proof}
The proof is given in \cite[Theorem 5, p. 110]{Widder}.
\end{proof}

\vskip 0.3cm

\begin{remark}
Condition \eqref{1.15} implies that there exists a constant $C$ such that
\beq
|\mu_n| < \frac{C}{(n+1)^{\frac{p-1}{p}}} \qquad (n=0,1,2,\ldots;\ p>1).
\label{1.16a}
\eeq 
This bound is easily obtained by observing that from definition \eqref{1.10} we have:
$w_n^{(n)}\{\mu\}=\mu_n$; 
now, condition \eqref{1.15} requires that, for every $n\in\N$ and $k\leqslant n$, 
$\left|w_k^{(n)}\{\mu\}\right|^p < M/(n+1)^{(p-1)}$ which, evaluated for $k=n$, yields bound \eqref{1.16a}
with $C=M^{1/p}$.
\label{rem:222}
\end{remark}

\vskip 0.2cm

\begin{proposition}
\label{pro:4}
If both conditions required by Propositions \ref{pro:1} and \ref{pro:3} hold,
then the sequence $\{\mu_k\}_{k=0}^\infty$ can be represented by an integral of the form \eqref{1.16}
where $\varphi(t)$ enjoys the following properties:
\begin{itemize}\setlength\itemsep{0.4em}
\item[(i)] it belongs to $L^p(0,1)$; 
\item[(ii)] it represents the density function of a probability distribution
concentrated on the closed interval $[0,1]$.
\end{itemize}
\end{proposition}

\begin{proof}
Since the sequence $\{\mu_k\}_{k=0}^\infty$ satisfies both conditions \eqref{1.12} and \eqref{1.15}, and $\mu_0=1$,
then the statement of the proposition follows obviously.
\end{proof}

\vskip 0.3cm

\subsection{Hausdorff conditions and Hardy spaces}
\label{subse:hausdorff}

In the lemmas that follow we study the interpolation of sequences of numbers which, within our program
aiming at characterizing holomorphic extensions associated with series of functions, 
will be the coefficients of trigonometric series
(see Section \ref{subse:derivation}) and of power series (see Section \ref{subse:taylor}).

\begin{lemma}
\label{lem:1}
Assume that the sequence of numbers $\{g_k\}_{k=0}^\infty$ satisfies condition \eqref{1.15} of Proposition \ref{pro:3}
with $p=2+\varepsilon$ ($\varepsilon>0$ arbitrarily small). Then they are the restriction to the integers 
of the following Laplace transform:
\beq
\label{1.17}
\tg(\lambda) = \int_0^{+\infty} e^{-(\lambda+\ud) v}\,e^{-\frac{v}{2}} \, \varphi(e^{-v})\,\rmd v
\qquad \left(\lambda\in\C,\,\Real\lambda\geqslant -\frac{1}{2}\right),
\eeq
i.e., $\tg(\lambda)|_{\lambda=k}=g_k$ ($k=0,1,2,\ldots$), where
$e^{-v/2}\varphi(e^{-v}) \in L^1(0,+\infty) \cap L^2(0,+\infty)$, and
$\tg(\lambda)$ is the unique Carlsonian interpolation of the sequence $\{g_k\}_{k=0}^\infty$.
Moreover, $\tg(-\frac{1}{2}+\rmi\nu)$ ($\nu\in\R$) is a continuous function which tends to zero as $\nu\to\pm\infty$.
\end{lemma}

\begin{proof}
If the given sequence $\{g_k\}_{k=0}^\infty$ satisfies the Hausdorff condition \eqref{1.15} with $p=2+\varepsilon$
($\varepsilon>0$ arbitrarily small), then (see \eqref{1.16})
\beq
\label{1.18}
g_k = \int_0^1 t^k \varphi(t)\,\rmd t \qquad (k = 0,1,2,\ldots),
\eeq
where $\varphi(t) \in L^{(2+\varepsilon)}(0,1)$. Put now $t=e^{-v}$ in \eqref{1.18}; we get
\beq
\label{1.19}
g_k = \int_0^{+\infty} e^{-(k+\ud)v}\,e^{-\frac{v}{2}}\,\varphi(e^{-v})\,\rmd v
\qquad (k=0,1,2,\ldots). \nonumber
\eeq
Therefore, the numbers $\{g_k\}_{k=0}^\infty$ can be formally regarded as the restriction to the 
integers of the following Laplace transform:
\beq
\label{1.20}
\tg(\lambda) = \int_0^{+\infty} e^{-(\lambda+\ud)v}\,e^{-\frac{v}{2}}\,\varphi(e^{-v})\,\rmd v
\qquad \left(\Real\lambda>-\frac{1}{2}\right).
\eeq
In fact, $\tg(\lambda)|_{\lambda=k}=g_k$. Now,
$\int_0^{+\infty}\left|e^{-v/2}\varphi(e^{-v})\right|^2\,\rmd v
=\int_0^{1}\left|\varphi(t)\right|^2\,\rmd t < \infty$
since $\varphi(t)\in L^{(2+\varepsilon)}(0,1)$ and, \emph{a-fortiori}, $\varphi(t)$ belongs also to $L^2(0,1)$.
Then, in view of the Paley-Wiener theorem \cite{Hoffman}, we can say that $\tg(\lambda)\in H^2(\CCm)$,
which is the Hardy space introduced in Section \ref{se:introduction}. 
Use can then be made of the Carlson theorem \cite{Boas},
which guarantees that $\tg(\lambda)$ represents the unique Carlsonian interpolation 
of the sequence $\{g_k\}_{k=0}^\infty$.
We are left to show that $e^{-v/2}\varphi(e^{-v})$ belongs also to $L^1(0,+\infty)$. 
This amounts to proving that
\beq
\int_0^{+\infty}\left|e^{-v/2}\varphi(e^{-v})\right|\,\rmd v 
= \int_0^{1}\left|\varphi(t)/\sqrt{t}\right|\,\rmd t < \infty. \nonumber
\label{add1}
\eeq
By H\"older's inequality, we have
\beq
\label{1.22}
\int_0^1 \frac{|\varphi(t)|}{\sqrt{t}}\,\rmd t
\leqslant\left(\int_0^1\left|\varphi(t)\right|^{2+\varepsilon}\,\rmd t\right)^{\frac{1}{(2+\varepsilon)}}
\!\!\left(\int_0^1 t^{-\frac{2+\varepsilon}{2+2\varepsilon}}\,\rmd t\right)^{\frac{(1+\varepsilon)}{(2+\varepsilon)}}
\!\!<\infty, \nonumber
\eeq
where the rightmost integral converges since $\frac{2+\varepsilon}{2+2\varepsilon}<1$ for any $\varepsilon>0$, and
$\varphi\in L^{(2+\varepsilon)}(0,1)$; it is then proved that $e^{-v/2}\varphi(e^{-v})\in L^1(0,+\infty)$.
Therefore, equality \eqref{1.20} may be extended up to $\Real\lambda=-\frac{1}{2}$.
Finally, setting $\lambda=-\frac{1}{2}+\rmi\nu$ ($\nu\in\R$) in \eqref{1.20}, we obtain
\beq
\tg\left(-\frac{1}{2}+\rmi\nu\right)=\int_0^{+\infty}e^{-\rmi\nu v}\left[e^{-\frac{v}{2}}
\varphi(e^{-v})\right]\,\rmd v
=\cF\left\{h(v)e^{-\frac{v}{2}}\varphi(e^{-v})\right\} \quad (\nu\in\R),
\label{1.22.1}
\nonumber
\eeq
where $h(v)$ is the Heaviside step function,
and $\cF$ denotes the Fourier integral operator. Since $e^{-v/2}\varphi(e^{-v})\in L^1(0,+\infty)$,
the Riemann-Lebesgue theorem guarantees that $\tg(-\frac{1}{2}+\rmi\nu)$ ($\nu\in\R$) is a continuous function,
which tends to zero as $\nu\to\pm\infty$.
\end{proof}

\begin{lemma}
\label{lem:2}
Assume that the sequence of numbers $\{f_k\}_{k=0}^\infty$, with $f_k \doteq (k+1) g_k$, 
satisfies condition \eqref{1.15}
of Proposition \ref{pro:3} with $p=2+\varepsilon$ ($\varepsilon>0$ arbitrarily small). Then:
\begin{itemize}\setlength\itemsep{0.4em}
\item[(i)] There exists a unique Carlsonian interpolation $\tg(\lambda)$ ($\Real\lambda\geqslant -\frac{1}{2}$)
of the sequence $\{g_k\}_{k=0}^\infty$, 
which is holomorphic in the half-plane $\Real\lambda>-\frac{1}{2}$.
\item[(ii)] $\lambda\tg(\lambda)\in L^2(-\infty,+\infty)$ for any fixed value of $\Real\lambda\geqslant-\frac{1}{2}$.
\item[(iii)] $\lambda\tg(\lambda)$ tends uniformly to zero as $\lambda\to\infty$ inside any fixed half-plane of the form
$\Real\lambda\geqslant\delta>-\frac{1}{2}$.
\item[(iv)] $\tg(\lambda)\in H^2\left(\CCm\right)$.
\item[(v)] $\tg(\lambda) \in L^1(-\infty,+\infty)$ for any fixed value of $\Real\lambda\geqslant-\frac{1}{2}$.
\item[(vi)] $\tg(-\frac{1}{2}+\rmi\nu)$ is a continuous function which tends to zero as $\nu\to\pm\infty$.
\item[(vii)] $\sup_{\ell>-1/2,\scriptstyle\nu\in\R}\left|\tg(\ell+\rmi\nu)\right|=\left|\tg\left(-\frac{1}{2}+\rmi\nu\right)\right|$.
\end{itemize}
\end{lemma}

\begin{proof}
By Lemma \ref{lem:1}, the numbers $f_k$ are the restriction to the integers of the Laplace transform
\beq
\tf(\lambda) = \int_0^{+\infty} e^{-(\lambda+\ud)v}\,e^{-\frac{v}{2}}\,\psi(e^{-v})\,\rmd v,
\qquad
\left(\Real\lambda\geqslant -\frac{1}{2}\right),
\label{popo:1}
\eeq
where $e^{-v/2}\,\psi(e^{-v})\in L^1(0,+\infty) \cap L^2(0,+\infty)$.
Then, $\tf(\lambda)\in H^2(\CCm)$ (hence, holomorphic in $\CCm$) and, accordingly,
$\tf(\ell+\rmi\nu) \in L^2(-\infty,+\infty)$ ($\ell\doteq\Real\lambda,\nu\doteq\Imag\lambda$) for any fixed
value of $\ell\geqslant -\frac{1}{2}$. Finally, $\tf(\lambda)$ tends uniformly to zero as $\lambda\to\infty$
inside any fixed half-plane $\Real\lambda\geqslant\delta>-\frac{1}{2}$ (see \cite{Hoffman}).
Now, we define $\tg(\lambda)\doteq\tf(\lambda)/(\lambda+1)$, which is holomorphic in the space $\CCm$ and Carlsonian
since $\tf(\lambda)$ is. Evidently, the restriction to the integers of $\tg(\lambda)$ is:
$\tg(\lambda)|_{\lambda=k}=\tf(k)/(k+1)=f_k/(k+1)=g_k$, ($k=0,1,2,\ldots$), and therefore
we can conclude that $\tg(\lambda)$ is the unique Carlsonian interpolation of the
sequence $\{g_k\}_{k=0}^\infty$, and statement (i) is proved.
Statements (ii) and (iii) follow immediately from the properties of $\tf(\lambda)=(\lambda+1)\tg(\lambda)$ 
given above. \\
Since $\tf(\lambda)\in H^2(\CCm)$ we can say that 
$\int_{-\infty}^{+\infty}\left|\tf(\ell+\rmi\nu)\right|^2\,\rmd\nu$ is bounded, for $\ell>-\frac{1}{2}$, by a
constant $M$ independent of $\ell$. Then we have
\beq
\int_{-\infty}^{+\infty}\left|\tg(\ell+\rmi\nu)\right|^2\,\rmd\nu\leqslant
4\int_{-\infty}^{+\infty}\left|\tf(\ell+\rmi\nu)\right|^2\,\rmd\nu\leqslant 4M, \nonumber
\eeq
which proves statement (iv). 
Regarding statement (v), by using the Schwarz inequality we have
\beq
\begin{split}
& \int_{-\infty}^{+\infty}\left|\tg(\ell+\rmi\nu)\right|\,\rmd\nu=
\int_{-\infty}^{+\infty}\left|\frac{\tf(\ell+\rmi\nu)}{(\ell+1)+\rmi\nu}\right|\,\rmd\nu \\
&\quad\leqslant
\left(\int_{-\infty}^{+\infty}\frac{1}{|(\ell+1)+\rmi\nu|^2}\,\rmd\nu\right)^{\ud}
\left(\int_{-\infty}^{+\infty}|\tf(\ell+\rmi\nu)|^2\,\rmd\nu\right)^{\ud}<\infty
\quad \left(\ell\geqslant-\frac{1}{2}\right).
\end{split}
\label{1.22.2}
\nonumber
\eeq
In order to prove point (vi) we consider \eqref{popo:1} for $\Real\lambda=-\frac{1}{2}$:
\beq
\begin{split}
\tf\left(-\frac{1}{2}+\rmi\nu\right)&=\left(\frac{1}{2}+\rmi\nu\right)\tg\left(-\frac{1}{2}+\rmi\nu\right) \\
&=\int_0^{+\infty}e^{-\rmi\nu v}e^{-\frac{v}{2}}\psi(e^{-v})\,\rmd v
=\cF\left\{h(v)e^{-\frac{v}{2}}\psi(e^{-v})\right\}<\infty.
\end{split}
\label{1.22.4}
\nonumber
\eeq
By the Riemann-Lebesgue theorem it follows that
$\tg\left(-\frac{1}{2}+\rmi\nu\right)$ ($\nu\in\R$) is a continuous function tending to zero as $\nu\to\pm\infty$,
and statement (vi) then follows. 
Finally, in order to prove statement (vii), we note that the Laplace transform \eqref{popo:1} holds also for 
$\Real\lambda=-\frac{1}{2}$ since $e^{-v/2}\psi(e^{-v})\in L^1(-\infty,+\infty)\cap L^2(-\infty,+\infty)$. 
Then,
\beq 
\sup_{\staccrel{\ell>-1/2}{\scriptstyle\nu\in\R}}\left|\tf(\ell+\rmi\nu)\right|
=\left|\tf(-\frac{1}{2}+\rmi\nu)\right|. \nonumber
\label{add2}
\eeq
Therefore, recalling that $\tg(\lambda)=\tf(\lambda)/(\lambda+1)$, we have
\beq
\begin{split}
\sup_{\ell>-1/2,{\scriptstyle\nu\in\R}}\left|\tg(\ell+\rmi\nu)\right|=
\sup_{\ell>-1/2,{\scriptstyle\nu\in\R}}\frac{\left|\tf(\ell+\rmi\nu)\right|}{|\ell+1+\rmi\nu|}
=\frac{\left|\tf\left(-\frac{1}{2}+\rmi\nu\right)\right|}{\sqrt{\nu^2+\frac{1}{4}}}
=\left|\tg\left(-\frac{1}{2}+\rmi\nu\right)\right|.
\label{coo:3}
\end{split}
\nonumber
\eeq
\end{proof}

\vspace*{0.2cm}

\begin{remark}
In the next section we shall study the holomorphic extension associated with trigonometric series of the form
$\frac{1}{2\pi}\sum_{k=1}^\infty g_k \exp(-\rmi k\tau)$ in order to derive the KMS analytic structure 
of the thermal Green functions from properties of the sequence $\{g_k\}$.
Thus, we shall be led to consider sequences $\{f_k\}_{k=1}^\infty$ ($f_k=kg_k$), i.e., 
starting from $k=1$ up to $\infty$, instead of from $k=0$. 
We may equivalently consider the enlarged sequence $\{f_k\}_{k=0}^\infty$ ($f_k=kg_k$), i.e., 
starting from $k=0$, with $f_0=0$ and $g_0$ arbitrary. Assuming that condition \eqref{1.15} 
of Proposition \ref{pro:3} with $p=2$
is satisfied by the sequence of numbers $\{f_{k}\}_{k=0}^\infty$, we can write
$f_{k}=\int_0^1 t^{k} \theta(t)\,\rmd t$, ($k=0,1,2,\ldots$), with $\theta(t) \in L^2(0,1)$.
Then, by using arguments close to those followed in the previous lemmas, the numbers $f_k$ 
are the restriction to the non-negative integers of the Laplace transform
\beq
\begin{split}
& \tf(\lambda)=\int_0^{+\infty} e^{-(\lambda+\ud)v}\,e^{-\frac{v}{2}}\,\theta(e^{-v})\,\rmd v
\quad \left(\Real\lambda>-\frac{1}{2}\right); \\
& \tf(\lambda)|_{\lambda=k}=f_k,
\qquad (k=0,1,2,\ldots).
\end{split}
\label{KL.31}
\eeq
However, since we shall be dealing with only the numbers $f_k$, with $k=1,2,\ldots$
then \eqref{KL.31} needs to hold only for $\Real\lambda\geqslant\frac{1}{2}$.
Therefore, it may be rewritten as
\beq
\begin{split}
& \tf(\lambda)=\int_0^{+\infty} e^{-(\lambda-\ud)v}\,e^{-\frac{3v}{2}}\,\theta(e^{-v})\,\rmd v
\quad \left(\Real\lambda\geqslant\frac{1}{2}\right); \\
& \tf(\lambda)|_{\lambda=k}=f_k \quad (k=1,2,3,\ldots),
\end{split}
\label{KL.31bis}
\eeq
where the function $e^{-3v/2}\theta(e^{-v})$ is easily seen to belong to $L^1(0,+\infty)\cap L^2(0,+\infty)$.
Setting now $\Real\lambda=\frac{1}{2}$ in \eqref{KL.31bis}, we have
\beq
\tf\left(\frac{1}{2}+\rmi\nu\right)=\int_0^{+\infty} e^{-\rmi\nu v} e^{-\frac{3v}{2}} \theta(e^{-v})\,\rmd v
= \cF\left\{h(v)e^{-\frac{3v}{2}}\theta(e^{-v})\right\}
\qquad (\nu\in\R),
\label{KL.33}
\nonumber
\eeq
which, by the Riemann-Lebesgue theorem, allows us to state that $\tf(\frac{1}{2}+\rmi\nu)$ ($\nu\in\R$)
is a continuous function which tends to zero as $\nu\to\pm\infty$.
\label{rem:1}
\end{remark}

\vspace*{0.4cm}

\noindent
Then, we can prove the following lemma.

\begin{lemma}
\label{lem:3}
Assume that the sequence $\{f_k\}_{k=0}^\infty$, with $f_k \doteq k g_k$,
satisfies condition \eqref{1.15} of Proposition \ref{pro:3} with $p=2$. Then:
\begin{itemize}\setlength\itemsep{0.5em}
\item[(i)] There exists a unique Carlsonian interpolation $\tg(\lambda)$ ($\Real\lambda\geqslant\frac{1}{2}$) 
of the coefficients $g_k$ ($k=1,2,\ldots$), which is holomorphic in the half-plane $\Real\lambda>\frac{1}{2}$.
\item[(ii)] $\lambda\tg(\lambda)\in L^2(-\infty,+\infty)$ for any fixed value of $\Real\lambda\geqslant\frac{1}{2}$.
\item[(iii)] $\lambda\tg(\lambda)$ tends uniformly to zero as $\lambda\rightarrow\infty$ inside any fixed half-plane of the form
$\Real\lambda\geqslant\delta>\frac{1}{2}$.
\item[(iv)] $\tg(\lambda)\in H^2\left(\CCp\right)$.
\item[(v)] $\tg(\lambda)\in L^1(-\infty,+\infty)$ for any fixed value of $\Real\lambda\geqslant\frac{1}{2}$.
\item[(vi)] $\tg(\frac{1}{2}+\rmi\nu)$ is a continuous function which tends to zero as $\nu\to\pm\infty$.
\item[(vii)] $\sup_{\ell> 1/2,{\scriptstyle\nu\in\R}}\left|\tg(\ell+\rmi\nu)\right|
=|\tg(\frac{1}{2}+\rmi\nu)|$.
\end{itemize}
\end{lemma}

\begin{proof}
Since $\{f_k\}_{k=0}^\infty$ satisfies condition \eqref{1.15} with $p=2$,
and taking into account what stated in Remark \ref{rem:1},
statements (i)-(vii) can be proved by following arguments analogous to those used in the proof
of points (i)-(vii) of Lemma \ref{lem:2}. 
\end{proof}

\noindent
\begin{examples}
We now give a few simple examples which aim at illustrating the properties proved above.
Consider the sequences (with $k=0,1,2,\ldots$):
\begin{align}
\qquad & g_k = \int_0^1 t^k (t-t^2)\,\rmd t = \frac{1}{(k+2)(k+3)},
\label{1.27} \\
\qquad & f_k = \int_0^1 t^k (2t^2-t)\,\rmd t = \frac{k+1}{(k+2)(k+3)}=(k+1) \, g_k, 
\label{1.28} \\
\qquad & \og_k = \int_0^1 t^k \,\rmd t = \frac{1}{(k+1)}.
\label{1.29}
\end{align}
Putting $t=e^{-v}$ in \eqref{1.27}, we obtain:
\beq
g_k = \int_0^{+\infty} e^{-(k+\ud)v}\,e^{-\frac{v}{2}}\,(e^{-v}-e^{-2v})\,\rmd v \qquad (k=0,1,2,\ldots).
\label{1.30}
\nonumber
\eeq
The sequence $\{g_k\}_{k=0}^\infty$ can be regarded as the restriction to the integers of the
following Laplace transform:
\beq
\tg(\lambda) = \int_0^{+\infty} e^{-(\lambda+\ud)v}\,e^{-\frac{v}{2}}\,(e^{-v}-e^{-2v})\,\rmd v 
\qquad \left(\Real\lambda\geqslant-\frac{1}{2}\right),
\label{1.31}
\nonumber
\eeq
where the function $g(v)=e^{-v/2}(e^{-v}-e^{-2v})$ belongs to
$L^1(0,+\infty)\cap L^2(0,+\infty)$. Obviously $g(v)$ tends to zero as $v\to +\infty$, and $g(0)=0$.

Analogously, in \eqref{1.28} we see that the numbers $f_k$ can be regarded as the restriction to the
integers of the Laplace transform
\beq
\tf(\lambda)=\int_0^{+\infty} e^{-(\lambda+\ud)v}\,e^{-\frac{v}{2}}\,(2e^{-2v}-e^{-v})\,\rmd v 
\qquad \left(\Real\lambda\geqslant-\frac{1}{2}\right),
\label{1.32}
\nonumber
\eeq
where again $f(v)=e^{-v/2}(2e^{-2v}-e^{-v})$ belongs to
$L^1(0,+\infty)\cap L^2(0,+\infty)$. Evidently, $f(v)$ tends to zero as $v\to +\infty$, but in this case $f(0)=1$.

Finally, consider the example in \eqref{1.29}. Again the sequence $\{\og_k\}_{k=0}^\infty$ can be
regarded as the restriction to the integers of the following Laplace transform:
\beq
\tog(\lambda) = \int_0^{+\infty} e^{-(\lambda+\ud)v}\,e^{-\frac{v}{2}}\,\rmd v 
\qquad \left(\Real\lambda\geqslant-\frac{1}{2}\right).
\label{1.32.2}
\nonumber
\eeq
The function $\og(v)=e^{-v/2} \in L^1(0,+\infty) \cap L^2(0,+\infty)$, and $\og(0)=1$.
\end{examples}

\vskip 0.4cm

\section{Holomorphic extension associated with trigonometric series and the KMS analytic structure}
\label{se:kms}

\subsection{Holomorphic extension associated with trigonometric series}
\label{subse:derivation}

In Thermal Quantum Field Theory the following problem plays a key role (see the Appendix).

\vskip 0.3cm

\begin{problem}
\label{prob:1}
Given the trigonometric series
\beq
\cG(\tau;\cdot)=\frac{1}{\beta}\sum_{n=-\infty}^{+\infty}g(\zeta_n;\cdot)\, e^{-\rmi\zeta_n\tau}
\quad \!\!\left(
\begin{aligned}
& \zeta_n=n\,\frac{\pi}{\beta}, \ n= \mathrm{even\ or\ odd\ integer;} \\
& \beta=(\mathrm{Temperature})^{-1}; \,\tau=(u+\rmi v)\in\C
\end{aligned}
\right),
\label{4.1}
\eeq
find the conditions on the coefficients $g(\zeta_n;\cdot)$ (where the dot denotes the
extra ``spectator variables''; see the Appendix), which are sufficient to guarantee that the
function $\cG(\tau;\cdot)$ satisfies the KMS analytic structure, i.e.,
\begin{itemize}
\item[(a)] it is analytic in the open strips $k\beta<u<(k+1)\,\beta$ \, ($v\in\R$, $k\in\Z$),
and continuous at the boundaries;
\item[(b)] it is periodic (resp., antiperiodic) for bosons (resp., fermions) with period $\beta$:
\beq
\cG(\tau+\beta;\cdot)=
\left\{
\begin{array}{rl}
\cG(\tau;\cdot) & \mathrm{for~bosons;~} (\tau\in\C), \\[+5pt]
-\cG(\tau;\cdot) & \mathrm{for~fermions;~} (\tau\in\C).
\end{array}
\right.
\label{4.2}
\eeq
\end{itemize}
\end{problem}

We introduce in the complex plane of the variable $\tau=u+\rmi v$ ($u,v\in\R$)
the following domains:
the half-planes $\cI_\pm\doteq\{\tau\in\C\,:\,\Imag\tau\gtrless 0\}$, and
the cut-domains $\cI_+\setminus\Xi_+^{(\beta)}$ and $\cI_-\setminus\Xi_-^{(\beta)}$,
where the cuts are given by
\beq
\Xi_\pm^{(\beta)}\doteq\{\tau\in\C\,:\,\tau=k\beta\pm\rmi v;\, v>0, k\in\Z\}. \nonumber
\label{add3}
\eeq
In this context, and following the standard Matsubara approach, the variable $v$ represents the
real time (see the Appendix).
Moreover, we denote by $\dot{A}$ any subset of $\C$, which is invariant under the translation
by $k\beta$, $k\in\Z$. Accordingly, the cut $\tau$-plane
$\C\setminus\left(\dot{\Xi}_+^{(\beta)} \cup \dot{\Xi}_-^{(\beta)}\right)$ will be denoted by $\dot{\Pi}^{(\beta)}$.

Let us begin by considering a system of bosons; the $\beta$-periodicity of $\cG(\tau;\cdot)$ (see \eqref{4.2}) 
is satisfied by putting $\zeta_n=n\,\frac{\pi}{\beta}$ with $n=2k$, $k\in\Z$, i.e.,
\beq
\cG^{(\rmb)}(\tau)=\frac{1}{\beta}\sum_{k=-\infty}^{+\infty}g\left(\frac{2k\pi}{\beta}\right)\, 
e^{-\rmi k\frac{2\pi}{\beta}\tau},
\label{4.1-boson}
\eeq
where the superscript `$\rmb$' stands for recalling that we are considering a system of bosons
(notice that, in order to make the notation lighter, from now on we omit to write the ``spectator variables'').
Now, by re-scaling suitably the variable $\tau$: i.e, $\frac{2\pi}{\beta}\tau\to\tau$, so that the period $\beta$ 
is re-scaled to $2\pi$, and defining $g_k^{(\rmb)} \doteq \frac{2\pi}{\beta}\, g(k)$ for $k=0, \pm 1,\pm 2,\ldots,$
we are left to analyze the $2\pi$-periodic function
\beq
\cG^{(\rmb)}(\tau)=\frac{1}{2\pi}\sum_{k=-\infty}^{+\infty}g_k^{(\rmb)}\, e^{-\rmi k\tau}. \nonumber
\label{add4}
\eeq
Finally, it is convenient to write the latter series as
\beq
\cG^{(\rmb)}(\tau)=\cG^{(+;\rmb)}(\tau)+\cG^{(-;\rmb)}(\tau)+g^{(\rmb)}_0/2\pi, \nonumber
\label{add5}
\eeq
where:
\begin{subequations}
\label{4.4}
\begin{align}
\cG^{(+;\rmb)}(\tau) =& \frac{1}{2\pi}\sum_{k=1}^{\infty} g_k^{(+;\rmb)}\,e^{-\rmi k\tau};
\qquad
g_k^{(+;\rmb)} \doteq g_{k}^{(\rmb)} \quad (k=1,2,\ldots), \label{4.4a} \\[+6pt]
\cG^{(-;\rmb)}(\tau) =& \frac{1}{2\pi}\sum_{k=-1}^{-\infty} g_k^{(\rmb)}\,e^{-\rmi k\tau}
=\frac{1}{2\pi}\sum_{k=1}^{\infty} g_k^{(-;\rmb)}\,e^{\rmi k\tau};
\qquad
g_k^{(-;\rmb)} \doteq g_{-k}^{(\rmb)} \quad (k=1,2,\ldots). \label{4.4b}
\end{align}
\end{subequations}
Now, we can prove the following theorem stating the properties of the function $\cG^{(+;\rmb)}(\tau)$
(similar results hold for $\cG^{(-;\rmb)}(\tau)$).

\vskip 0.3cm

\begin{theorem}
\label{the:6}
Suppose that the sequence of numbers $\{f_k\}_{k=0}^\infty$, with $f_k \doteq k g_k^{(+;\rmb)}$,
satisfies condition \eqref{1.15} of Proposition \ref{pro:3} with $p=2$. Then the trigonometric 
series in \eqref{4.4a} enjoys the following properties:
\begin{enumerate}\setlength\itemsep{0.4em}
\item It converges uniformly on any compact subdomain of the half-plane $\cI_-$ to a function
$\cG^{(+;\rmb)}(\tau)$ holomorphic in $\cI_-$, continuous on the axis $v\doteq\Imag\tau=0$.
\item $\cG^{(+;\rmb)}(\tau)$ admits a holomorphic extension to $\cI_+\setminus\dot{\Xi}_+^{(2\pi)}$.
\item The jump function $J^{(+;\rmb)}(v)$, which represents the discontinuity of $~\rmi\cG^{(+;\rmb)}(\tau)$ 
across the cut $\dot{\Xi}_+^{(2\pi)}$ (i.e., the difference between the values on the upper and the lower lip 
of the cut),
\beq
\begin{split}
J^{(+;\rmb)}(v) \doteq \rmi\lim_{\staccrel{\varepsilon\to 0}{\scriptstyle\varepsilon>0}}
& \left[\cG^{(+;\rmb)}(2k\pi+\varepsilon+\rmi v) - \cG^{(+;\rmb)}(2k\pi-\varepsilon+\rmi v)\right] \\
& \hspace{4.5cm} (k\in\Z,v\in\R^+),
\end{split}
\label{4.6}
\eeq
is a continuous function satysfying the following bound:
\beq
\left|J^{(+;\rmb)}(v)\right|\leqslant
\left\|\tg^{(+;\rmb)}_\ell\right\|_1\,e^{\ell v}
\qquad \left(\ell\geqslant\frac{1}{2}; v\in\R^+\right),
\label{4.5}
\eeq
where
\beq
\left\|\tg^{(+;\rmb)}_\ell\right\|_1 \doteq
\frac{1}{2\pi}\int_{-\infty}^{+\infty}\left|\tg^{(+;\rmb)}(\ell+\rmi\nu)\right|\,\rmd\nu<\infty
\qquad \left(\ell\geqslant\frac{1}{2}\right),
\label{4.7}
\eeq
and $\tg^{(+;\rmb)}(\lambda)$ ($\lambda=\ell+\rmi\nu$, $\ell\geqslant \ud$, $\nu\in\R$)
is the unique Carlsonian interpolation of the Fourier coefficients
$\{g^{(+;\rmb)}_k\}_{k=1}^\infty$.
\item $J^{(+;\rmb)}(v) = o(e^{v/2})$ as $v\to+\infty$, and $J^{(+;\rmb)}(0) = 0$.
\item The jump function can be represented by the following transform:
\beq
J^{(+;\rmb)}(v)=
\frac{e^{\frac{v}{2}}}{2\pi}\int_{-\infty}^{+\infty}
\tg^{(+;\rmb)}\left(\frac{1}{2}+\rmi\nu\right)e^{\rmi\nu v}\,\rmd\nu
\qquad (v\in\R^+).
\label{4.7.bis}
\eeq
\item The Laplace transform of the jump function across the cut, $\tJ^{(+;\rmb)}(\lambda)$, is given by:
\beq
\tJ^{(+;\rmb)}(\lambda) \doteq \int_0^{+\infty} J^{(+;\rmb)}(v) e^{-\lambda v}\,\rmd v
= \tg^{(+;\rmb)}(\lambda)\qquad \left(\Real\lambda>\frac{1}{2}\right),
\label{4.8}
\eeq
which is holomorphic in the half-plane $\Real\lambda>\frac{1}{2}$.
\item The following Plancherel-type equality holds true:
\beq
\int_{-\infty}^{+\infty}\left|\tg^{(+;\rmb)}(\ell+\rmi\nu)\right|^2\,\rmd\nu=
2\pi\int_{0}^{+\infty}\left|J^{(+;\rmb)}(v) e^{-\ell v}\right|^2\,\rmd v
\qquad \left(\ell\geqslant\frac{1}{2}\right).
\label{4.9}
\nonumber
\eeq
\end{enumerate}
\end{theorem}

\begin{proof}
Since the sequence $\{f_k\}_{k=0}^\infty$ ($f_k=kg^{(+;\rmb)}_k$) satisfies condition \eqref{1.15}
with $p=2$, then, given an arbitrary constant $C$, there exists an integer $k_0$ such that,
for $k>k_0$, $|g^{(+;\rmb)}_k|\leqslant C$ (see \eqref{1.16a}). Accordingly, we have
\beq
\left|\frac{1}{2\pi}\sum_{k=k_0+1}^\infty g^{(+;\rmb)}_k e^{-\rmi k\tau}\right|
\leqslant\frac{C}{2\pi}\sum_{k=k_0+1}^\infty e^{kv}
\qquad (v\equiv\Imag\tau;~  C=\mathrm{constant}).
\label{TT:1}
\eeq
The series on the r.h.s. of \eqref{TT:1} converges uniformly on any compact subdomain
contained in the half-plane $v<0$, i.e., for $v\leqslant v_0<0$.
Since we can write:
\beq
\frac{1}{2\pi}\sum_{k=1}^\infty g^{(+;\rmb)}_k e^{-\rmi k\tau}
=\frac{1}{2\pi}\sum_{k=k_0+1}^\infty g^{(+;\rmb)}_k e^{-\rmi k\tau} + T_{k_0}(\tau),
\label{T:2}
\nonumber
\eeq
where $T_{k_0}(\tau)$ is a trigonometric polynomial,
we can say, by Weierstrass' theorem on uniformly convergent series of analytic functions,
that $\frac{1}{2\pi}\sum_{k=1}^\infty g^{(+;\rmb)}_k e^{-\rmi k\tau}$ converges uniformly 
on any compact subdomain of $\cI_-$ to a function $\cG^{(+;\rmb)}(\tau)$ holomorphic in $\cI_-$.
Furthermore, since $kg^{(+;\rmb)}_k\staccrel{\longrightarrow}{k\to+\infty}0$,
given an arbitrary constant $C'$,
there exists an integer $k_1$ such that for $v=0$ we have:
\beq
\left|\frac{1}{2\pi}\sum_{k=k_1}^\infty g^{(+;\rmb)}_k e^{-\rmi ku}\right|
\leqslant\frac{1}{2\pi}\sum_{k=k_1}^\infty \left|g^{(+;\rmb)}_k\right|
\leqslant C' \sum_{k=k_1}^\infty \frac{1}{k^{1+\varepsilon}} < \infty
\qquad (\varepsilon>0).
\label{T:3}
\nonumber
\eeq
Then, applying the Weierstrass theorem on uniformly convergent series of continuous
functions, we can say that the series converges to a continuous function. Statement (1) is thus proved.
\begin{figure}[tb]
\begin{center}
\leavevmode
\includegraphics[width=12.5cm]{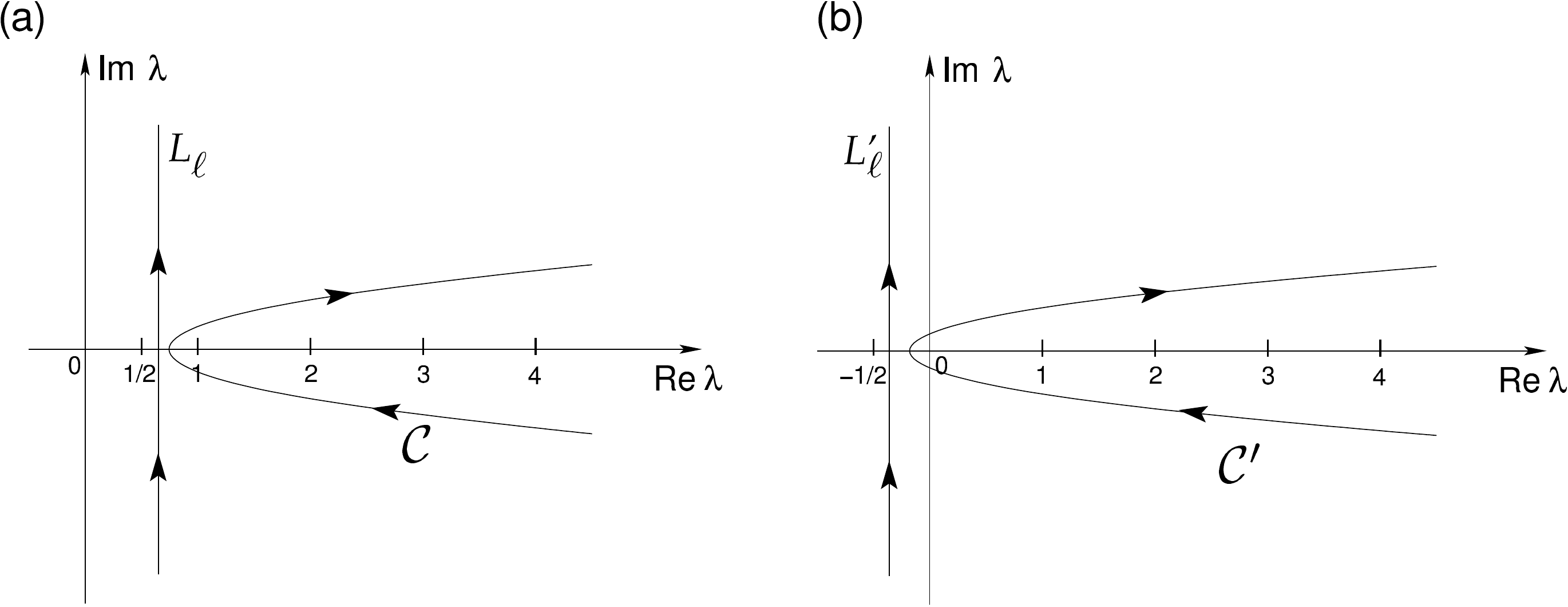}
\caption{\label{fig:1}
\small
(a) and (b): Integration paths of integrals (\protect\ref{TT:4}) and (\protect\ref{Q:1}), respectively.}
\end{center}
\end{figure}

Regarding the point (2), we write the following integral:
\beq
\cG_\eta^{(+;\rmb)}(u) \doteq
\frac{\rmi}{4\pi}\int_{\cC}\tg^{(+;\rmb)}(\lambda)\,
\frac{e^{-\rmi\lambda(u-\eta\pi)}}{\sin\pi\lambda}\,\rmd\lambda
\qquad (\eta = \pm),
\label{TT:4}
\eeq
where $\tg^{(+;\rmb)}(\lambda)$
is the unique Carlsonian interpolation of the coefficients $\{g_k^{(+;\rmb)}\}_{k=1}^\infty$,
which is holomorphic in the half-plane $\Real\lambda>\frac{1}{2}$ (see statement
(i) of Lemma \ref{lem:3}). The contour $\cC$ is contained in the half-plane $\CCp$,
encircles the real semiaxis $\Real\lambda\geqslant\frac{1}{2}$, and is chosen to cross the latter
in a point $\ell\not\in\N$ (see Figure \ref{fig:1}a). Consider now the term
$\frac{\exp[-\rmi\lambda(u-\eta\pi)]}{\sin\pi\lambda}$; 
the following inequalities hold ($\lambda=\ell+\rmi\nu$):
\beq
\left|e^{-\rmi(\ell+\rmi\nu)(u-\eta\pi)}\right|
\leqslant 2\cosh\pi\nu
\quad \mathrm{for} \quad
\left\{
\begin{array}{cc}
0 \leqslant u \leqslant 2\pi & ~ \mathrm{if} \quad \eta = +, \\
-2\pi \leqslant u \leqslant 0 & ~ \mathrm{if} \quad \eta = -,
\end{array}
\right.
\label{TT:5}
\eeq
\beq
\left|\sin\pi(\ell+\rmi\nu)\right| \geqslant \sinh\pi\nu,
\label{T:6}
\eeq
and
\beq
\left|\sin\pi(\ell+\rmi\nu)\right|\geqslant\left|\sin\pi\ell\right|\,\cosh\pi\nu.
\label{T:8}
\eeq
Therefore, from \eqref{TT:5} and \eqref{T:6} we have
\beq
\left|\frac{e^{-\rmi(\ell+\rmi\nu)(u-\eta\pi)}}{\sin\pi(\ell+\rmi\nu)}\right|
\leqslant 2\left|\frac{\cosh\pi\nu}{\sinh\pi\nu}\right|
\quad \mathrm{for} \quad
\left\{
\begin{array}{cc}
0 \leqslant u \leqslant 2\pi & ~ \mathrm{if} \quad \eta = +, \\
-2\pi \leqslant u \leqslant 0 & ~ \mathrm{if} \quad \eta = -,
\end{array}
\right.
\label{T:7}
\eeq
while, combining \eqref{TT:5} and \eqref{T:8}, we have:
\beq
\left|\frac{e^{-\rmi(\ell+\rmi\nu)(u-\eta\pi)}}{\sin\pi(\ell+\rmi\nu)}\right|
\leqslant\frac{2}{|\sin\pi\ell|}<\infty
\quad \mathrm{for} \,
\left\{
\begin{array}{cc}
0 \leqslant u \leqslant 2\pi & ~ \mathrm{if} \quad \eta = +, \\
-2\pi \leqslant u \leqslant 0 & ~ \mathrm{if} \quad \eta = -,
\end{array}
\right.
\,\mathrm{and} \ \ell\not\in\N.
\label{T:9}
\eeq
Integral \eqref{TT:4} converges since $\lambda\,\tg^{(+;\rmb)}(\lambda)$ tends uniformly
to zero as $\lambda$ tends to infinity in any fixed half-plane $\Real\lambda\geqslant\delta>\frac{1}{2}$
(see statement (iii) of Lemma \ref{lem:3}), and
furthermore, the term $\left|\frac{\exp[-\rmi(\ell+\rmi\nu)(u-\eta\pi)]}{\sin\pi(\ell+\rmi\nu)}\right|$
is bounded by a constant for $\ell\not\in\N$ (see \eqref{T:9}), and is bounded by $2$ as $\nu\to\pm\infty$
in view of inequality \eqref{T:7}.
Next, the contour $\cC$ can be distorted and replaced by a line $L_\ell$ parallel to the imaginary axis and
crossing the real axis at $\Real\lambda=\ell$
($\ell\not\in\N, \ell\geqslant\frac{1}{2}$) (see Figure \ref{fig:1}a), provided the real variable $u$ is kept 
in $[0,2\pi]$ for $\cG_+^{(+;\rmb)}(u)$ and in $[-2\pi,0]$ for $\cG_-^{(+;\rmb)}(u)$, respectively. 
Note that the integral along $L_\ell$ converges since $\tg^{(+;\rmb)}(\lambda) \in L^1(-\infty,+\infty)$ 
for any fixed value of $\Real\lambda\geqslant\frac{1}{2}$ (see statement (v) of Lemma \ref{lem:3})
and in view of inequality \eqref{T:9}.
We may now apply the Watson resummation method to integral \eqref{TT:4}.
For $u\in[0,2\pi]$ we obtain:
\beq
\frac{\rmi}{4\pi}\int_{\cC}\tg^{(+;\rmb)}(\lambda)\, \frac{e^{-\rmi\lambda(u-\pi)}}{\sin\pi\lambda}\,\rmd\lambda
=\frac{1}{2\pi}\sum_{k=1}^\infty g_k^{(+;\rmb)} e^{-\rmi ku}
\qquad \left(0\leqslant u\leqslant 2\pi\right),
\label{T:10}
\eeq
where the integration path $\cC$ crosses the real axis at a point $\ell$ with $\ud\leqslant\ell<1$.
Distorting the contour $\cC$ into the line $L_\ell$, which is admissible
as explained above, we obtain the following representation for the function $\cG_+^{(+;\rmb)}(u)$:
\beq
\begin{split}
\cG_+^{(+;\rmb)}(u) &= -\frac{1}{4\pi}\int_{-\infty}^{+\infty}\tg^{(+;\rmb)}(\ell+\rmi\nu)
\frac{e^{-\rmi(\ell+\rmi\nu)(u-\pi)}}{\sin\pi(\ell+\rmi\nu)}\,\rmd\nu \\
&= \frac{1}{2\pi}\sum_{k=1}^\infty g_k^{(+;\rmb)} e^{-\rmi ku}
\qquad\left(0\leqslant u\leqslant 2\pi; \, \frac{1}{2}\leqslant\ell<1 \right).
\end{split}
\label{T:11}
\eeq
Analogously, for $u\in[-2\pi,0]$ we obtain:
\beq
\begin{split}
\cG_-^{(+;\rmb)}(u) &= -\frac{1}{4\pi}\int_{-\infty}^{+\infty}\tg^{(+;\rmb)}(\ell+\rmi\nu)
\frac{e^{-\rmi(\ell+\rmi\nu)(u+\pi)}}{\sin\pi(\ell+\rmi\nu)}\,\rmd\nu \\
&= \frac{1}{2\pi}\sum_{k=1}^\infty g_k^{(+;\rmb)} e^{-\rmi ku}
\qquad\left(-2\pi\leqslant u \leqslant 0, \,\frac{1}{2}\leqslant\ell<1\right).
\end{split}
\label{T:12}
\eeq
We now substitute into the integral in \eqref{T:11} the real variable $u$ with the complex variable $\tau=u+\rmi v$,
and the resulting integral can be proved to provide an analytic
continuation of $\cG_+^{(+;\rmb)}(u)$ in the strip $0<u<2\pi$, $v\in\R^+$, continuous
in the closure of the latter. In fact, from the first equality in \eqref{T:11} we formally obtain
\beq
\cG_+^{(+;\rmb)}(u+\rmi v) =
\frac{1}{2\pi}\,e^{\ell v}\,\int_{-\infty}^{+\infty} H_\ell^u(\nu)\,e^{\rmi\nu v}\,\rmd\nu
\qquad \left(0\leqslant u \leqslant 2\pi;\, \frac{1}{2}\leqslant\ell<1\right),
\label{T:13}
\eeq
where
\beq
H_\ell^u(\nu) =
-\frac{\tg^{(+;\rmb)}(\ell+\rmi\nu)e^{-\rmi(\ell+\rmi\nu)(u-\pi)}}{2\sin\pi(\ell+\rmi\nu)}.
\label{T:14}
\eeq
In view of inequality \eqref{T:9} and statement (vii) of Lemma \ref{lem:3} we have:
\beq
\left|H_\ell^u(\nu)\right|
\leqslant
\frac{|\tg^{(+;\rmb)}(\frac{1}{2}+\rmi\nu)|}{|\sin\pi\ell|}
\qquad \left(0\leqslant u \leqslant 2\pi; \ell\geqslant\frac{1}{2},\ell\not\in\N\right),
\label{T:14bis}
\eeq
which, along with statement (v) of Lemma \ref{lem:3}, guarantees $H_\ell^u(\nu)\in L^1(-\infty,+\infty)$
for $0\leqslant u \leqslant 2\pi, \ell\geqslant\frac{1}{2},\ell\not\in\N$.
Therefore, formulae \eqref{T:13}, \eqref{T:14}, and \eqref{T:14bis} define $\cG_+^{(+;\rmb)}(\tau)$ as an analytic
continuation of $\cG_+^{(+;\rmb)}(u)$ in the strip
$\{\tau=u+\rmi v\,:\,0<u<2\pi,v\in\R^+\}$, continuous on the closure of the latter in view of
the Riemann-Lebesgue theorem.
Proceeding analogously, we can obtain an analytic continuation
of $\cG_-^{(+;\rmb)}(u)$ in the strip $\{\tau=u+\rmi v\,:\,-2\pi<u<0,v\in\R^+\}$, 
continuous on the closure of the latter.
Therefore, the function $\cG^{(+;\rmb)}(\tau)$ admits a holomorphic
extension to the cut-domain $\cI_+\setminus\dot{\Xi}_+^{(2\pi)}$, and statement (2) is thus proved.

The discontinuity $J^{(+;\rmb)}(v)$ of $\cG^{(+;\rmb)}(\tau)$ across the cut at $u=0$, 
which is defined in \eqref{4.6}
(for the $2\pi$-periodicity of $\cG^{(+;\rmb)}(\tau)$ we may consider only the cut at $u=0$), equals
$\rmi[\cG_+^{(+;\rmb)}(\rmi v)-\cG_-^{(+;\rmb)}(\rmi v)]$  ($v\in\R^+$). 
It can be computed by replacing $u$ by $\rmi v$
into the integrals in \eqref{T:11} and \eqref{T:12}, and then subtracting Eq. \eqref{T:12}
from Eq. \eqref{T:11} side by side:
\beq
\begin{split}
J^{(+;\rmb)}(v)=\rmi\left[\cG_+^{(+;\rmb)}(\rmi v)-\cG_-^{(+;\rmb)}(\rmi v)\right] 
& =\frac{1}{2\pi}\int_{-\infty}^{+\infty}\tg^{(+;\rmb)}(\ell+\rmi\nu)e^{(\ell+\rmi\nu)v}\,\rmd\nu \\
&\qquad\qquad\qquad\qquad \left(v\in\R^+, \ell\geqslant\frac{1}{2}\right),
\end{split}
\label{T:16}
\eeq
which yields
\beq
\left|J^{(+;\rmb)}(v)\right|\leqslant\left\|\tg^{(+;\rmb)}_\ell\right\|_1\,e^{\ell v}
\qquad \left(v\in\R^+,\ell\geqslant\frac{1}{2}\right),
\label{T:19}
\nonumber
\eeq
where $\left\|\tg^{(+;\rmb)}_\ell\right\|_1$, defined in \eqref{4.7}, is guaranteed to be finite 
by statement (v) of Lemma \ref{lem:3}; statement (3) is thus proved.
From \eqref{T:16} we have
\beq
J^{(+;\rmb)}(v)e^{-\ell v}=\frac{1}{2\pi}\int_{-\infty}^{+\infty}\tg^{(+;\rmb)}(\ell+\rmi\nu)\,e^{\rmi\nu v}\,\rmd\nu
\qquad \left(v\in\R^+,\ell\geqslant\frac{1}{2}\right).
\label{T:21}
\eeq
By the Riemann-Lebesgue theorem it follows that $J^{(+;\rmb)}(v)e^{-\ell v}$
is a continuous function tending to zero as $v\to+\infty$, and, therefore,
$J^{(+;\rmb)}(v)$ is a continuous function of $v$ ($v\in\R^+$)
with $J^{(+;\rmb)}(v) = o(e^{v/2})$ as $v\to+\infty$. Moreover,
$J^{(+;\rmb)}(0)=0$ for the continuity of $\cG^{(+;\rmb)}(\tau)$ on the real axis.
Statement (4) is thus proved. Putting $\ell=\frac{1}{2}$ in formula \eqref{T:21},
statement (5) follows. Next, inverting \eqref{T:21}, we have:
\beq
\tg^{(+;\rmb)}(\ell+\rmi\nu)=
\int_{0}^{+\infty}J^{(+;\rmb)}(v)\,e^{-(\ell+\rmi\nu) v}\,\rmd v
\qquad \left(\ell>\frac{1}{2}\right),
\label{T:22}
\nonumber
\eeq
where the integral on the r.h.s. converges for $\ell>\frac{1}{2}$.
It defines $\tJ^{(+;\rmb)}(\lambda)$, that is, the Laplace transform of $J^{(+;\rmb)}(v)$,
holomorphic in the half-plane $\Real\lambda>\frac{1}{2}$, and statement (6) follows.
Finally, recalling that $\tg^{(+;\rmb)}(\ell+\rmi\nu)$ ($\nu\in\R$) belongs to $L^2(-\infty,+\infty)$ 
for any fixed value
$\ell\geqslant\frac{1}{2}$ (see statements (ii) and (vi) of Lemma \ref{lem:3}),
we obtain the Plancherel equality:
\beq
\int_{-\infty}^{+\infty}\left|\tg^{(+;\rmb)}\left(\ell+\rmi\nu\right)\right|^2\,\rmd\nu=
2\pi\int_0^{+\infty}\left|J^{(+;\rmb)}(v)\,e^{-\ell v}\right|^2\rmd v \qquad \left(\ell\geqslant\frac{1}{2}\right),
\label{T:23}
\nonumber
\eeq
which proves statement (7).
\end{proof}

\vskip 0.2cm

We can now state the following theorem.

\begin{theorem}
\label{the:7}
Suppose that the sequence of numbers $\{f_k\}_{k=0}^\infty$, with $f_k \doteq k g_k^{(-;\rmb)}$,
satisfies condition \eqref{1.15} of Proposition \ref{pro:3} with $p=2$. Then the trigonometric 
series in \eqref{4.4b} enjoys the following properties:
\begin{enumerate}\setlength\itemsep{0.4em}
\item[(1')] It converges uniformly on any compact subdomain of $\cI_+$ to a function
$\cG^{(-;\rmb)}(\tau)$ holomorphic in $\cI_+$, continuous on the axis $\Imag\tau=0$.
\item[(2')] $\cG^{(-;\rmb)}(\tau)$ admits a holomorphic extension to $\cI_-\setminus\dot{\Xi}_-^{(2\pi)}$.
\item[(3')] The jump function across the cut $\dot{\Xi}_-^{(2\pi)}$, i.e.,
\beq
\begin{split}
J^{(-;\rmb)}(v) \doteq
-\rmi\lim_{\staccrel{\varepsilon\to 0}{\scriptstyle\varepsilon>0}}
& \left[\cG^{(-;\rmb)}(2k\pi+\varepsilon+\rmi v)-\cG^{(-;\rmb)} (2k\pi-\varepsilon+\rmi v)\right] \\
& \hspace{4.5cm}(k\in\Z,v\in\R^-),
\end{split}
\nonumber
\label{4.12}
\eeq
is a continuous function satysfying the following bound:
\beq
\left|J^{(-;\rmb)}(v)\right|\leqslant
\left\|\tg^{(-;\rmb)}_\ell\right\|_1\,e^{-\ell v}
\qquad \left(\ell\geqslant\frac{1}{2}; v\in\R^-\right), \nonumber
\label{4.11}
\eeq
where:
\beq
\left\|\tg^{(-;\rmb)}_\ell\right\|_1 \doteq
\frac{1}{2\pi}\int_{-\infty}^{+\infty}\left|\tg^{(-;\rmb)}(\ell+\rmi\nu)\right|\,\rmd\nu<\infty
\qquad \left(\ell\geqslant\frac{1}{2}\right), \nonumber
\label{4.13}
\eeq
and $\tg^{(-;\rmb)}(\lambda)$ ($\lambda=\ell+\rmi\nu$, $\ell\geqslant \ud$, $\nu\in\R$)
is the unique Carlsonian interpolation of the Fourier coefficients
$\{g^{(-;\rmb)}_k\}_{k=1}^\infty$.
\item[(4')] $J^{(-;\rmb)}(v) = o(e^{-v/2})$ as $v\to-\infty$, and $J^{(-;\rmb)}(0) = 0$.
\item[(5')] The jump function can be represented by the following transform:
\beq
J^{(-;\rmb)}(v)=
\frac{e^{-\frac{v}{2}}}{2\pi}\int_{-\infty}^{+\infty}
\tg^{(-;\rmb)}\left(\frac{1}{2}+\rmi\nu\right)e^{-\rmi\nu v}\,\rmd\nu
\qquad (v\in\R^-). \nonumber
\label{4.13.bis}
\eeq
\item[(6')] The Laplace transform of the jump function across the cut, $\tJ^{(-;\rmb)}(\lambda)$, is given by:
\beq
\tJ^{(-;\rmb)}(\lambda) \doteq \int_{-\infty}^0 J^{(-;\rmb)}(v) e^{\lambda v}\,\rmd v
= \tg^{(-;\rmb)}(\lambda)\qquad \left(\Real\lambda>\frac{1}{2}\right),
\label{4.14}
\eeq
which is holomorphic in the half-plane $\Real\lambda>\frac{1}{2}$.
\item[(7')] The following Plancherel-type equality holds true:
\beq
\int_{-\infty}^{+\infty}\left|\tg^{(-;\rmb)}(\ell+\rmi\nu)\right|^2\,\rmd\nu=
2\pi\int_{-\infty}^0\left|J^{(-;\rmb)}(v) e^{\ell v}\right|^2\,\rmd v
\qquad \left(\ell\geqslant\frac{1}{2}\right).
\label{4.15}
\nonumber
\eeq
\end{enumerate}
\end{theorem}

\begin{proof}
The proof is analogous to that of Theorem \ref{the:6}.
\end{proof}

\vskip 0.3cm

\subsection{Derivation of the KMS analytic structure}
\label{subse:derKMS}

We can now collect the results of Theorems \ref{the:6} and \ref{the:7},
and derive the KMS analytic structure for the series \eqref{4.1}
in the case of bosons (and temperature scaled to have $\beta=2\pi$;
the generalization to any value of $\beta$ is straightforward).
For this purpose we prove the following theorem.

\vskip 0.2cm

\begin{theorem}
\label{the:8}
If the Fourier coefficients $g_k^{(\pm;\rmb)}$ satisfy the conditions of Theorems
\ref{the:6} and \ref{the:7}, then:
\begin{itemize}\setlength\itemsep{0.4em}
\item[(i)] The functions $\cG^{(\pm;\rmb)}(u)$ ($u\in\R$) are continuous.
\item[(ii)] The function $\cG^{(+;\rmb)}(\tau)$ (resp., $\cG^{(-;\rmb)}(\tau)$)
is holomorphic in $\cI_-$ (resp., $\cI_+$), and admits a holomorphic extension to
$\cI_+\setminus\dot{\Xi}_+^{(2\pi)}$ (resp., $\cI_-\setminus\dot{\Xi}_-^{(2\pi)}$).
\item[(iii)] The jump function $J^{(+;\rmb)}(v)$ across the cut $\dot{\Xi}_+^{(2\pi)}$ 
is continuous on $[0,+\infty)$, $J^{(+;\rmb)}(v)=o(e^{v/2})$ as $v\to+\infty$, and $J^{(+;\rmb)}(0)=0$. \\
The jump function $J^{(-;\rmb)}(v)$ across the cut $\dot{\Xi}_-^{(2\pi)}$ 
is continuous on $(-\infty,0]$, $J^{(-;\rmb)}(v)=o(e^{-v/2})$ as $v\to-\infty$, and $J^{(-;\rmb)}(0)=0$.
\item[(iv)] The function $\cG^{(\rmb)}(\tau)=\cG^{(+;\rmb)}(\tau)+\cG^{(-;\rmb)}(\tau)+g_0^{(\rmb)}/2\pi$ 
$(\tau=u+\rmi v$; $u,v\in\R$) is periodic with period $2\pi$, analytic in the strips
$\{\tau=u+\rmi v\,:\,2k\pi<u<2(k+1)\pi, v\in\R; k\in\Z\}$, and continuous up to the boundaries.
\item[(v)] The Laplace transforms $\tJ^{(+;\rmb)}(\lambda)$
and $\tJ^{(-;\rmb)}(\lambda)$ of the respective jump functions
$J^{(+;\rmb)}(v)$ and $J^{(-;\rmb)}(v)$ 
are analytic in the half-plane $\Real\lambda>\frac{1}{2}$.
\item[(vi)] The Fourier coefficients $g^{(+;\rmb)}_k$
($k=1,2,\ldots$) are the restriction to the positive integers of the Laplace transform
$\tJ^{(+;\rmb)}(\lambda)$, i.e., $g^{(+;\rmb)}_k=\tJ^{(+;\rmb)}(k)$, ($k=1,2,3,\ldots$). \\
The Fourier coefficients $g^{(-;\rmb)}_k$
($k=1,2,\ldots$) are the restriction to the positive integers of the Laplace transform
$\tJ^{(-;\rmb)}(\lambda)$, i.e., $g^{(-;\rmb)}_k=\tJ^{(-;\rmb)}(k)$, ($k=1,2,\ldots$).
\end{itemize}
\end{theorem}

\begin{proof}
Statements (i), (ii), (iii), and (v) follow directly from Theorems \ref{the:6} and \ref{the:7}.
Regarding statement (iv), the $2\pi$-periodicity of $\cG^{(\rmb)}(\tau)$ (see \eqref{4.1-boson})
follows from the one of $\cG^{(\pm;\rmb)}(\tau)$.
Next, in Theorems \ref{the:6} and \ref{the:7}, we have proved that
$\cG^{(\rmb)}(\tau)$ is analytic in the half-strips
$\{\tau=u+\rmi v\,:\,2k\pi<u<2(k+1)\pi,v>0;k\in\Z\}$ and $\{\tau=u+\rmi v\,:\,2k\pi<u<2(k+1)\pi,v<0;k\in\Z\}$; 
moreover, the function $\cG^{(\rmb)}(u)=\cG^{(+;\rmb)}(u)+\cG^{(-;\rmb)}(u)+g_0^{(\rmb)}/2\pi$ is continuous.
By Schwarz's reflection principle, it follows that $\cG^{(\rmb)}(\tau)$
is holomorphic in the strips $\{\tau=u+\rmi v\,:\,2k\pi<u<2(k+1)\pi,v\in\R;k\in\Z\}$.
Furthermore, taking into account that $\cG^{(\rmb)}(\tau)$
is continuous on the closure of the strips where it is analytic (as proved
in Theorems \ref{the:6} and \ref{the:7}), and that it is continuous at $\tau=u$, statement (iv) follows.
Finally, statement (vi) follows from \eqref{4.8} and \eqref{4.14},
recalling that $\tg^{(+;\rmb)}(\lambda)$ (resp., $\tg^{(-;\rmb)}(\lambda)$) is the unique Carlsonian 
interpolation of the Fourier coefficients $g^{(+;\rmb)}_k$ (resp., $g^{(-;\rmb)}_k$), that is, 
$\tg^{(\pm;\rmb)}(\lambda)|_{\lambda=k}=g^{(\pm;\rmb)}_k$ ($k= 1, 2,\ldots$).
\end{proof}

\vskip 0.3cm

Consider now a system of fermions.
The $\beta$-antiperiodicity of $\cG^{(\rmf)}(\tau)$ (see condition \eqref{4.2}) is satisfied
by setting into expansion \eqref{4.1} $\zeta_n=n\frac{\pi}{\beta}$, $n=(2k+1)$, $k\in\Z$.
Then, following arguments similar to those used in the case of bosons, namely,
re-scaling $\tau$ by a factor $\frac{2\pi}{\beta}$ and defining $g_k^{(\rmf)}\doteq\frac{2\pi}{\beta}g(k+\ud)$
for $k\in\Z$, series \eqref{4.1} can be written as:
$\cG^{(\rmf)}(\tau)=\cG^{(+;\rmf)}(\tau)+\cG^{(-;\rmf)}(\tau)+\frac{1}{2\pi}g_0^{(\rmf)}e^{-\rmi\frac{\tau}{2}}$
($\tau\in\C$), where:
\begin{align*}
& \cG^{(+;\rmf)}(\tau) = 
\frac{1}{2\pi}\sum_{k=1}^{+\infty}g_k^{(\rmf)}\,e^{-\rmi(k+\frac{1}{2})\tau}
= \frac{e^{-\rmi\frac{\tau}{2}}}{2\pi}\sum_{k=1}^{+\infty}g_k^{(+;\rmf)}\,e^{-\rmi k\tau};
\quad
g_k^{(+;\rmf)} \doteq g_{k}^{(\rmf)},
\\[+5pt]
& \cG^{(-;\rmf)}(\tau) = \frac{1}{2\pi}
\sum_{k=-1}^{-\infty}g_k^{(\rmf)}\,e^{-\rmi(k+\frac{1}{2})\tau}
= \frac{e^{-\rmi\frac{\tau}{2}}}{2\pi}\sum_{k=1}^{+\infty}g_k^{(-;\rmf)}\,e^{\rmi k\tau};
\quad
g_k^{(-;\rmf)} \doteq g_{-k}^{(\rmf)}.
\end{align*}
Proceeding as in the proofs of Theorems \ref{the:6}, \ref{the:7}, and \ref{the:8} 
(with only slight modifications), we can achieve results analogous to those obtained 
in the case of bosons.

\vskip 0.4cm

\begin{remarks}
\label{rem:2}
(i) In Theorem \ref{the:6} (resp., \ref{the:7}) we have proved that the Laplace transforms
$\tJ^{(+;\rmb)}(\lambda)$ and $\tJ^{(-;\rmb)}(\lambda)$ are analytic in the half-plane 
$\Real\lambda>\frac{1}{2}$. These results follow from our assumptions on
the Fourier coefficients, which lead the jump function across the cut (e.g., $J^{(+;\rmb)}(v)$)
to satisfy an exponential-type bound of the form $C e^{\ell v}$ ($C=\mathrm{constant}$), 
for $\ell\geqslant\frac{1}{2}$, as expressed, e.g., by formula \eqref{4.5}.
If, instead, these discontinuity functions are assumed to satisfy a slow-growth condition of power-type
(i.e., they are \emph{tempered functions}), as done in \cite{Cuniberti}, then the analyticity domain 
can be enlarged to have $\tJ^{(\pm;\rmb)}(\lambda)$ holomorphic in the half-plane $\Real\lambda>0$, 
which is more realistic for physical systems.

(ii) Honesty requires to notice that to date we are not able to exhibit an example of physical system
whose thermal Green functions satisfy the conditions of Theorems
\ref{the:6}, \ref{the:7} and \ref{the:8}.
Work in order to overcome this limit is in progress.
\end{remarks}

\newpage

\section{Holomorphic extension associated with power series and connection between 
the jump function across the cut and the density of a probability distribution}
\label{se:taylor}

\subsection{Holomorphic extension associated with power series}
\label{subse:taylor}

We prove the following theorem.

\begin{theorem}
\label{the:1}
Consider the power series
\beq
\label{1.33}
\frac{1}{2\pi}\sum_{k=0}^\infty g_k z^k
\qquad (z\in\C; z=x+\rmi y; x,y\in\R),
\eeq
and suppose that the sequence of numbers $\{f_k\}_{k=0}^\infty$, with $f_k = (k+1) g_k$, 
satisfies condition \eqref{1.15} of Proposition \ref{pro:3} with $p=2+\varepsilon$.
Then:
\begin{itemize}\setlength\itemsep{0.4em}
\item[(i)] The series \eqref{1.33} converges uniformly on any compact subdomain of the open
unit disk $|z|<1$ to a function $G(z)$ holomorphic therein and continuous up to the unit circle $|z|=1$.
\item[(ii)] The function $G(z)$ admits a holomorphic extension to the cut-plane $\SSS$.
\item[(iii)] The jump function across the cut
$J(x) \doteq -\rmi[G_+(x)-G_-(x)]$, $x\in[1,+\infty)$, where
$G_\pm(x)\doteq\lim_{y\to 0;y>0}G(x\pm\rmi y)$, is a continuous
function, which satisfies the following bound:
\beq
\label{1.35}
|J(x)|\leqslant\left\|\tg_\ell\right\|_1 x^\ell
\qquad \left(\ell\geqslant-\frac{1}{2}; x\in[1,+\infty)\right),
\nonumber
\eeq
where
\beq
\left\|\tg_\ell\right\|_1\doteq\frac{1}{2\pi}\int_{-\infty}^{+\infty}\left|\tg(\ell+\rmi\nu)\right|\,\rmd\nu<\infty
\qquad \left(\ell\geqslant-\frac{1}{2}\right),
\label{1.36}
\eeq
and $\tg(\lambda)$ ($\Real\lambda\geqslant-\frac{1}{2}$) is the unique
Carlsonian interpolation of the coefficients $\{g_k\}_{k=0}^\infty$.
\item[(iv)] $J(x)=o(x^{-\ud})$ as $x\to+\infty$, and $J(1)=0$.
\item[(v)] The jump function $J(x)$ can be represented by the following inverse Mellin transform:
\beq
J(x)=\frac{1}{2\pi}\int_{-\infty}^{+\infty}\tg\left(\ell+\rmi\nu\right)\,x^{(\ell+\rmi\nu)}\,\rmd\nu
\qquad \left(\ell\geqslant-\frac{1}{2}; x\in[1,+\infty)\right).
\label{1.36.bis}
\nonumber
\eeq
\item[(vi)] The Mellin transform of the discontinuity function across the cut is
\beq
\hJ(\lambda)
\doteq\int_{1}^{+\infty} J(x)\,x^{-\lambda-1}\,\rmd x = \tg(\lambda)
\qquad \left(\Real\lambda>-\frac{1}{2}\right),
\label{1.37}
\nonumber
\eeq
and $\hJ(\lambda)$ is holomorphic in the half-plane $\Real\lambda>-\frac{1}{2}$.
\item[(vii)] The Plancherel formula associated with the Mellin transform reads
\beq
\int_{-\infty}^{+\infty}\left|\tg\left(\ell+\rmi\nu\right)\right|^2\,\rmd\nu
= {2\pi}
\int_{1}^{+\infty} \left|J(x)\right|^2\,x^{-2\ell-1}\,\rmd x
\qquad \left(\ell\geqslant-\frac{1}{2}\right),
\label{1.38}
\nonumber
\eeq
and, in particular, for $\ell=-\ud$ we have
\beq
\int_{-\infty}^{+\infty}\left|\tg\left(-\frac{1}{2}+\rmi\nu\right)\right|^2\,\rmd\nu
= 2\pi\int_{1}^{+\infty} \left|J(x)\right|^2\,\rmd x.
\label{1.39}
\nonumber
\eeq
\end{itemize}
\end{theorem}

\begin{proof}
The mapping $z=e^{-\rmi\tau}$ (with inverse: $\tau=\rmi\ln z$) maps
the cut-plane $\{\tau\in\C\setminus\dot{\Xi}_+^{(2\pi)}\}$ of the variable $\tau$
into the cut-plane $\SSS$ of the variable $z$; moreover,
it transforms the trigonometric series $\sum_{k=0}^\infty g_k e^{-\rmi k\tau}$
into the power series $\sum_{k=0}^\infty g_k z^k$. Therefore, the proof of the theorem coincides,
up to some modifications, with the proof of Theorem \ref{the:6}.
The main difference between the series being considered here
and that treated in Theorem \ref{the:6} is that now the series starts from $k=0$ instead of from $k=1$.
This difference is completely irrelevant for what concerns the proof of the statement (i), 
which coincides with that of statement (1) of Theorem \ref{the:6}, by just observing that the 
half-plane $\Imag\tau\leqslant 0$ is mapped onto the domain $|z|\leqslant 1$.
For what concerns the proof of the other statements,
in strict analogy with the procedure followed in the proof of Theorem \ref{the:6}, we perform
a Watson resummation of the series $\frac{1}{2\pi}\sum_{k=0}^\infty g_k e^{-\rmi ku}$, for $u\in[0,2\pi]$, 
as follows (cf. \eqref{T:10}):
\beq
\frac{1}{2\pi}\sum_{k=0}^\infty g_k e^{-\rmi ku}
=\frac{\rmi}{4\pi}
\int_{\cC'}\tg(\lambda)\frac{e^{-\rmi\lambda(u-\pi)}}{\sin\pi\lambda}\,\rmd\lambda
\qquad \left(0\leqslant u \leqslant 2\pi\right),
\label{Q:1}
\eeq
where the integration path $\cC'$ is contained in the half-plane $\CCm$,
encircles the semiaxis $\Real\lambda\geqslant-\frac{1}{2}$, and crosses it in a point $-\frac{1}{2}\leqslant\ell<0$ 
(see Figure \ref{fig:1}b).
Note that we may use all the statements of Lemma \ref{lem:2}, in particular the one which guarantees that
$\tg(\lambda)\in H^2(\CCm)$ represents the unique Carlsonian interpolation of the coefficients
$\{g_k\}_{k=0}^\infty$. Moreover, we know that $\lambda\tg(\lambda)\to 0$ as $\lambda\to\infty$ in any fixed
half-plane $\Real\lambda\geqslant\delta>-\frac{1}{2}$, and $\tg(\lambda)\in L^1(-\infty,+\infty)$
for any fixed value of $\Real\lambda\geqslant-\frac{1}{2}$. Next, proceeding as in the proof of Theorem \ref{the:6}
(up to a few obvious modifications), the contour $\cC'$ may be distorted suitably and replaced by the line $L_\ell'$
(see Figure \ref{fig:1}b); we then have
\beq
\begin{split}
\cG_+(u) &= -\frac{1}{4\pi}\int_{-\infty}^{+\infty}\frac{\tg(\ell+\rmi\nu)e^{-\rmi(\ell+\rmi\nu)(u-\pi)}}{\sin\pi(\ell+\rmi\nu)}\,\rmd\nu \\
&=\frac{1}{2\pi}\sum_{k=0}^\infty g_k e^{-\rmi ku}
\qquad \left(0\leqslant u \leqslant 2\pi;-\frac{1}{2}\leqslant\ell<0\right).
\end{split}
\label{Q:2}
\eeq
Analogously, for $-2\pi\leqslant u\leqslant 0$ we obtain:
\beq
\begin{split}
\cG_-(u) &= -\frac{1}{4\pi}\int_{-\infty}^{+\infty}\frac{\tg(\ell+\rmi\nu)e^{-\rmi(\ell+\rmi\nu)(u+\pi)}}{\sin\pi(\ell+\rmi\nu)}\,\rmd\nu \\
&=\frac{1}{2\pi}\sum_{k=0}^\infty g_k e^{-\rmi ku}
\qquad \left(-2\pi\leqslant u \leqslant 0;-\frac{1}{2}\leqslant\ell<0\right).
\end{split}
\label{Q:3}
\eeq
It is easy to show that both integrals in \eqref{Q:2} and \eqref{Q:3} converge
since $\tg(\ell+\rmi\nu)\in L^1(-\infty,+\infty)$ for any fixed value of $\ell\geqslant-\frac{1}{2}$
(see statement (v) of Lemma \ref{lem:2}), and $\tg(-\frac{1}{2}+\rmi\nu)$ is a continuous function tending to zero
as $\nu\to\pm\infty$ (statement (vi) of Lemma \ref{lem:2}). Furthermore, bounds
\eqref{T:8}, \eqref{T:7}, and \eqref{T:9} continue to hold in the present situation. Then, following
the arguments used in Theorem \ref{the:6}, we can prove that
$\cG_+(u)$ (resp., $\cG_-(u)$) admits an analytic continuation in the half-strip
$\{\tau=u+\rmi v\,:\,0<u<2\pi, v\in\R^+\}$ (resp., $\{\tau=u+\rmi v\,:\,-2\pi<u<0, v\in\R^+\}$), 
continuous up to the border of the strip.
Accordingly, $\cG(\tau)$ (we continue to use the notation of Theorem \ref{the:6})
admits a holomorphic extension to the cut-domain $\cI_+\setminus\dot{\Xi}_+^{(2\pi)}$.
Finally, replacing $u$ by $\rmi v$ in both the integrals in \eqref{Q:2} and \eqref{Q:3}
we obtain the following expression for the jump function $J(v)$:
\beq
J(v)\doteq\rmi
\left[\cG_+(\rmi v)-\cG_-(\rmi v)\right]=
\frac{1}{2\pi}\int_{-\infty}^{+\infty}\tg(\ell+\rmi\nu)e^{(\ell+\rmi\nu)v}\,\rmd\nu<\infty
\ \left(v\in\R^+;\ell\geqslant-\frac{1}{2}\right).
\label{Q:4}
\eeq
Therefore we have
\beq
|J(v)|\leqslant \|\tg_\ell\|_1\,e^{\ell v}
\qquad\left(v\in\R^+; \ell\geqslant-\frac{1}{2}\right),
\label{Q:5}
\nonumber
\eeq
with $\|\tg_\ell\|_1$ defined in \eqref{1.36}.
In view of the Riemann-Lebesgue theorem, from formula \eqref{Q:4} it follows that
$J(v)e^{-\ell v}$ ($\ell\geqslant-\frac{1}{2}$) is a continuous function tending
to zero as $v\to+\infty$. This implies that $J(v)=o(e^{-v/2})$ as $v\to+\infty$,
and, moreover, $J(0)=0$ since $\cG(u)$ ($u\in\R$) is a continuous function.
Next, inverting Eq. \eqref{Q:4}, we get
\beq
\tg(\lambda)=
\int_{0}^{+\infty} J(v) e^{-\lambda v}\,\rmd v \doteq \tJ(\lambda)
\qquad\left(\Real\lambda>-\frac{1}{2}\right),
\label{Q:7}
\nonumber
\eeq
where $\tJ(\lambda)$ is the Laplace transform of $J(v)$, which is holomorphic in the half-plane
$\Real\lambda>-\frac{1}{2}$.
Finally, from the Plancherel theorem it follows:
\beq
\int_{-\infty}^{+\infty}|\tg(\ell+\rmi\nu)|^2\,\rmd\nu
=2\pi\int_{0}^{+\infty}\left|J(v)e^{-\ell v}\right|^2\,\rmd v
\qquad\left(\ell\geqslant-\frac{1}{2}\right).
\label{Q:8}
\nonumber
\eeq
Let us now return to the mapping $z=e^{-\rmi\tau}$ from the complex $\tau$-plane
($\tau=u+\rmi v$; $u,v\in\R$) to the complex $z$-plane ($z=x+\rmi y$; $x,y\in\R$). We have:
$z=x+\rmi y=e^v(\cos u-\rmi\sin u)$. Therefore, the cut in the $\tau$-plane along the semiaxis 
$\{\tau=\rmi v, v\in\R^+\}$
is mapped into the cut in the $z$-plane along the semiaxis $\{z=x, x\in[1,+\infty)\}$;
precisely, the left (resp., right) lip of the cut at $\{\tau=\rmi v,v\in\R^+\}$ is mapped into the upper 
(resp., lower) lip of the cut at $\{z=x,x\in[1,+\infty)\}$.
Accordingly, the jump function in the $z$-plane which corresponds
to the $\tau$-plane jump function $J(v)=\rmi[\cG_+(\rmi v)-\cG_-(\rmi v)]$ is
(by a small abuse of notation we continue to denote it by $J$): $J(x)\doteq-\rmi[G_+(x)-G_-(x)]$,
where $G_\pm(x)\doteq\lim_{y\to 0;y>0}G(x\pm\rmi y)$, ($x\in[1,+\infty)$). Statements (iii) through (vii)
then follow immediately.
In Table \ref{tab:1} the results in the $\tau$-plane, referring to the trigonometric series
$\cG(\tau)=\frac{1}{2\pi}\sum_{k=0}^\infty g_k e^{-\rmi k \tau}$,
and the results in the $z$-plane referring to the power series \eqref{1.33}, are briefly summarized.
\end{proof}

\begin{table}[tb]
\caption{Comparison between the results in the $\tau$-plane and the $z$-plane.
\label{tab:1}
}
{\scriptsize
\begin{center}
\leavevmode
\begin{tabular}{c|c}
\hline \\
{\bf\small $\boldsymbol{\tau}$-plane} & {\bf\small $\boldsymbol{z}$-plane} \\ \\
\hline \\
\hspace{-0.8cm}\begin{minipage}{6.5cm}
\begin{itemize}
\item[1)] $\cG(\tau)$ is holomorphic in $\cI_-$, continuous on the boundary $\Imag\tau=0$.
\item[2)] $\cG(\tau)$ admits a holomorphic extension to $\cI_+\setminus\dot{\Xi}_+^{(2\pi)}$.
\item[3)] $J(v)\doteq\rmi[\cG_+(\rmi v)-\cG_-(\rmi v)]$ is continuous; $J(0)=0$;
$J(v)=o(e^{-v/2})$ as $v\to+\infty$.
\item[4)] $J(v)=\frac{1}{2\pi}\int_{-\infty}^{+\infty}\tg(\ell+\rmi\nu)e^{(\ell+\rmi\nu)v}\,\rmd \nu$,\\
($\ell\geqslant-\frac{1}{2},v\in\R^+$).
\item[5)] $|J(v)|\leqslant\|\tg_\ell\|_1 e^{\ell v}$, $\|\tg_\ell\|_1<\infty$, 
($\ell\geqslant-\frac{1}{2}$, $v\in\R^+$).
\item[6)]$\tJ(\lambda)=\int_0^{+\infty}J(v)e^{-\lambda v}\,\rmd v = \tg(\lambda)$,\\
($\Real\lambda>-\frac{1}{2}$).
\item[7)] $\int_{-\infty}^{+\infty}|\tg(\ell+\rmi\nu)|^2\,\rmd\nu$ \\
$\null\quad=2\pi\int_0^{+\infty}|J(v)e^{-\ell v}|^2\,\rmd v$, \ ($\ell\geqslant-\frac{1}{2}$).
\end{itemize}
\end{minipage}
&
\hspace{-0.6cm}\begin{minipage}{6.5cm}
\begin{itemize}
\item[i)] $G(z)$ is holomorphic in $|z|<1$, continuous on the boundary $|z|=1$.
\item[ii)] $G(z)$ admits a holomorphic extension to $\SSS$.
\item[iii)] $J(x)\doteq-\rmi[G_+(x)-G_-(x)]$ is continuous; $J(1)=0$;
$J(x)=o(x^{-1/2})$ as $x\to+\infty$.
\item[iv)] $J(x)=\frac{1}{2\pi}\int_{-\infty}^{+\infty}\tg(\ell+\rmi\nu)x^{(\ell+\rmi\nu)}\,\rmd\nu$,\\
($\ell\geqslant-\frac{1}{2},x\in[1,+\infty)$).
\item[v)] $|J(x)|\leqslant\|\tg_\ell\|_1 x^{\ell}$,~~
$|\tg_\ell\|_1<\infty$,~ ($\ell\geqslant-\frac{1}{2},$ \\ $x\in[1,+\infty)$).
\item[vi)]$\hJ(\lambda)=\int_1^{+\infty}J(x)x^{-\lambda-1}\,\rmd x=\tg(\lambda)$, \\
($\Real\lambda>-\frac{1}{2}$).
\item[vii)] $\int_{-\infty}^{+\infty}|\tg(\ell+\rmi\nu)|^2\,\rmd\nu$ \\
$\null\quad=2\pi\int_1^{+\infty}|J(x)|^2 x^{-2\ell-1}\,\rmd x$,\ ($\ell\geqslant-\frac{1}{2}$).
\end{itemize}
\end{minipage}
\\ \\ \hline
\end{tabular}
\end{center}
}
\vspace{0.5cm}
\end{table}

\subsection{Connection between the Laplace transform of the jump function and the coefficients of the power series}
\label{subse:connection}

We assume that the coefficients $\{g_k\}_{k=0}^\infty$
$G(z)=\frac{1}{2\pi}\sum_{k=0}^\infty g_k z^k$ ($z=x+\rmi y$) satisfy the conditions of Theorem \ref{the:1}.
Again, by the transformation $z=e^{-\rmi\tau}$ $(\tau=u+\rmi v)$, $x+\rmi y=e^v(\cos u-\rmi\sin u)$,
we may consider the trigonometric series $\cG(\tau)=\frac{1}{2\pi}\sum_{k=0}^\infty g_k e^{-\rmi k\tau}$.
As it has been proved in Subsection \ref{subse:taylor}, the function $\cG(\tau)$ is holomorphic
in the cut-plane $\dot{\Pi}_+^{(2\pi)} \doteq \C\setminus\dot{\Xi}_+^{(2\pi)}$,
where $\dot{\Xi}_+^{(2\pi)} \doteq\{\tau\in\C\,:\,\tau=2k\pi+\rmi v;\,k\in\Z,v>0\}$;
moreover, we denote by $\Xi_{+,0}\doteq\{\tau\in\C\,:\,\tau=\rmi v;v>0\}$ the cut at $u=0$.

\vskip 0.3cm 

Consider now the integral
\beq
\label{N:4}
I_\gamma (\lambda) \doteq
\int_\gamma \cG(\tau)\,e^{\rmi\lambda\tau}\,\rmd\tau \qquad \left(\lambda\in\CCm\right),
\eeq
$\CCm\doteq\{\lambda\in\C:\Real\lambda>-\frac{1}{2}\}$, with the prescription for the contour $\gamma$
shown in Figure \ref{fig:2}, i.e., $\gamma$ encircles the cut $\Xi_{+,0}$ and belongs to the half-strip 
$-\pi\leqslant u \leqslant\pi$,
$v\in\R^+$. Note that the integral \eqref{N:4} is well-defined for $\lambda\in\CCm$. Indeed,
from formula \eqref{Q:2} we can derive formulae analogous to equality \eqref{T:13}.
In fact, we can write (see also \eqref{Q:2}):
\beq
\label{N:5}
\begin{split}
e^{-\ell v}\cG_+(u+\rmi v) &=\frac{1}{2\pi}\int_{-\infty}^{+\infty} H_\ell^u(\nu)\,e^{\rmi\nu v}\,\rmd\nu \\
& \hspace{2.5cm} \left(0\leqslant u\leqslant 2\pi,v\in\R^+,\ell\geqslant-\frac{1}{2},\ell\not\in\N\right),
\end{split}
\eeq
where $\cG_+(u+\rmi v)$ is the analytic continuation of $\cG_+(u)$ in the strip $0<u< 2\pi, v\in\R^+$, and
\beq
\label{N:6}
H_\ell^u(\nu)=-\frac{\tg(\ell+\rmi\nu)e^{-\rmi(\ell+\rmi\nu)(u-\pi)}}{2\sin\pi(\ell+\rmi\nu)},
\eeq
with
\beq
\begin{split}
|\sin\pi\ell|e^{-\ell v}\left|\cG_+(u+\rmi v)\right|
& \leqslant\frac{1}{2\pi}\int_{-\infty}^{+\infty}\!\left|\tg\left(-\frac{1}{2}+\rmi\nu\right)\right|\,\rmd\nu <\infty \\
&\qquad\left(0\leqslant u\leqslant 2\pi,v\in\R^+,\ell\geqslant-\frac{1}{2},\ell\not\in\N\right).
\end{split}
\label{N:7}
\eeq
\begin{figure}[tb]
\begin{center}
\leavevmode
\includegraphics[width=11cm]{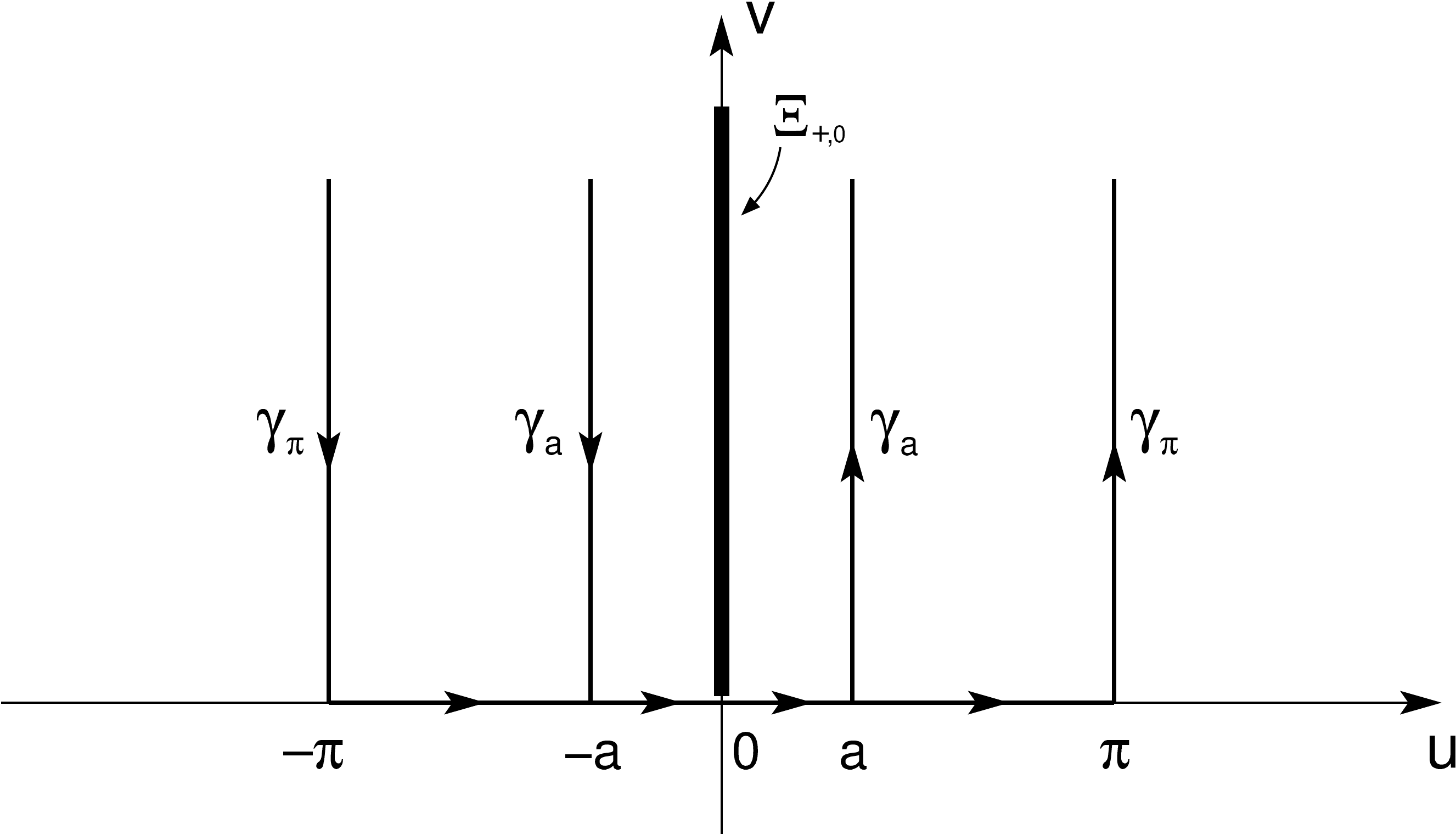}
\caption{\label{fig:2}
\small Integration paths.}
\end{center}
\end{figure}
Next, applying the Riemann-Lebesgue theorem, from formulae \eqref{N:5}, \eqref{N:6}, and \eqref{N:7},
it follows that $\lim_{v\to+\infty}\cG(u+\rmi v)e^{v/2}=0$ for $0<u<2\pi$.
Analogously, it can be proved that $\lim_{v\to+\infty}\cG(u+\rmi v)e^{v/2}=0$ for $-2\pi<u<0$. Indeed,
as we proved in the previous subsection, $\cG(\tau)$ is holomorphic in the half-strips
$-2\pi<u<0$ and $0<u<2\pi$, $v\in\R^+$, and admits continuous boundary values (from both sides) on the
cut $\Xi_{+,0}$. Furthermore, inequality \eqref{Q:4} and the Riemann-Lebesgue theorem guarantee
that $\lim_{v\to+\infty}J(v)e^{v/2}=0$, where $J(v)$ denotes the jump function across the cut $\Xi_{+,0}$.

According to the standard Cauchy distortion argument, $I_\gamma(\lambda)$ is independent of the path $\gamma$.
By choosing for $\gamma$ the path $\gamma_a$ (see Figure \ref{fig:2}), with support
$(-a+\rmi\infty,-a]\cup[-a,a]\cup[a,a+\rmi\infty)$ $(0<a<\pi)$, positively oriented,
we obtain as $a\to 0$:
\beq
\lim_{a\to 0}I_{\gamma_a}(\lambda)=\int_0^{+\infty} e^{-\lambda v}\,J(v)\rmd v =
(\cL J)(\lambda)
\qquad \left(\lambda\in\CCm\right),
\nonumber
\eeq
where $\cL$ denotes the Laplace transform operator, and $J(v)$ is the jump function across the
cut $\Xi_{+,0}$.
Next, choosing for $\gamma$ the path $\gamma_\pi$ (see Figure \ref{fig:2}), and exploiting the
$2\pi$-periodicity of the integrand $e^{\rmi\lambda\tau}\cG(\tau)$ when $\lambda=k$ is an integer, we can write:
\beq
\left. I_{\gamma_\pi}(\lambda)\right|_{\lambda=k}
= \int_{\gamma_\pi} e^{\rmi k\tau}\cG(\tau)\,\rmd\tau
=\int_{-\pi}^{+\pi}e^{\rmi ku}\cG(u)\,\rmd u = g_k
\qquad (k=0,1,2,\ldots),
\label{N:9}
\nonumber
\eeq
which are the Fourier coefficients of the expansion $\cG(u)=\frac{1}{2\pi}\sum_{k=0}^\infty g_k\,e^{-\rmi ku}$.

Applying once again the Cauchy theorem, we obtain the following equalities, analogous to those derived
in Section \ref{subse:derKMS} for the trigonometric series (see statement (vi) of Theorem \ref{the:8}):
\beq
\label{N:10}
\left.(\cL J)(\lambda)\right|_{\lambda=k} = g_k
\qquad (k=0,1,2,\ldots),
\nonumber
\eeq
which, finally, lead us to state the following proposition.

\begin{proposition}
\label{pro:5}
Consider the power series $\frac{1}{2\pi}\sum_{k=0}^\infty g_k z^k$, and suppose that all the
conditions required by Theorem \ref{the:1} are satisfied. Then the restriction to non-negative 
integers of the Laplace transform of the jump function across the cut coincides with the 
coefficients $\{g_k\}_{k=0}^\infty$.
\end{proposition}

\vspace*{0.2cm}

\subsection{Connection between the jump function across the cut and the density 
of a probability distribution}
\label{subse:connection-jump}

We begin by proving the following lemma.

\begin{lemma}
\label{lem:N1}
Consider the power series $\frac{1}{2\pi}\sum_{k=0}^\infty g_k z^k$ ($z\in\C$, $|z|<1$), and suppose
that the coefficients $\{g_k\}_{k=0}^\infty$ satisfy condition \eqref{1.15} of Proposition \ref{pro:3} with
$p=2+\varepsilon$ ($\varepsilon>0$ arbitrarily small). Moreover, in equality \eqref{1.16}
(which follows from \eqref{1.15}) we restrict the class of functions $\varphi(t)$ to non-negative functions:
$\varphi(t)\geqslant 0$. Then there exists a unique Carlsonian interpolation $\tg(\ell)$
($\ell\geqslant 0$) of the coefficients $g_k$ ($k=0,1,2,\ldots$) such that:
\begin{itemize}\setlength\itemsep{0.5em}
\item[(i)] $\sds (-1)^i \, \frac{\rmd^{i}\tg(\ell)}{\rmd\ell^{i}}\geqslant 0 \quad (i=0,1,2,\ldots)$, i.e., 
the function $\tg(\ell)$ ($\ell\geqslant 0$) is completely monotonic.
\item[(ii)] $\tg(\ell)$ can be represented by the following Laplace-Stieltjes transform:
\beq
\label{Q1}
\tg(\ell)=\int_0^{+\infty} e^{-\ell v}\,\rmd\nu(v) \qquad (\ell\geqslant 0),
\eeq
where $\nu(v)$ is a bounded measure on $\left([0,+\infty),\cB_{[0,+\infty)}\right)$,
$\cB_{[0,+\infty)}$ being the topological Borel field of $[0,+\infty)$. Moreover, $\nu(v)$ is an integral,
and $\rmd\nu(v)=e^{-v}\varphi(e^{-v})\,\rmd v$.
\end{itemize}
\end{lemma}

\begin{proof}
Since the set $\{g_k\}_{k=0}^\infty$ satisfies condition \eqref{1.15} of Proposition \ref{pro:3} 
with $p=2+\varepsilon$ (and, accordingly, in \eqref{1.16} the function $\varphi(t)$ belongs 
to $L^{2+\varepsilon}(0,1)$), then the coefficients $g_k$ ($k=0,1,2,\ldots$) admit a unique Carlsonian 
interpolation $\tg(\lambda)$, holomorphic in the half-plane
$\Real\lambda>-\frac{1}{2}$ and represented by formula \eqref{1.17}, where the function 
$e^{-v/2}\varphi(e^{-v})\in L^1(0+\infty)\cap L^2(0,+\infty)$ (see Lemma \ref{lem:1}). Formula \eqref{1.17} 
can be rewritten as follows
for real non-negative values of $\lambda$, i.e., on the semiaxis $\ell\geqslant 0$:
\beq
\label{Q2}
\tg(\ell) = \int_0^{+\infty} e^{-\ell v} e^{-v} \varphi(e^{-v})\,\rmd v \qquad (\ell\geqslant 0),
\eeq
from which we formally derive:
\beq
\label{Q3}
(-1)^i \, \frac{\rmd^{i}\tg(\ell)}{\rmd\ell^{i}}
= \int_0^{+\infty} e^{-\ell v} \, v^i \,e^{-v} \varphi(e^{-v})\,\rmd v \qquad (i=0,1,2,\ldots).
\nonumber
\eeq
It can be easily proved that $\left|\frac{\rmd^{i}\tg(\ell)}{\rmd\ell^{i}}\right|<\infty$ ($i=0,1,2,\ldots$) 
for the class of functions which satisfy the requirements imposed by \eqref{1.15} with $p=2+\varepsilon$. 
Now, if we assume that $\varphi(t)$ is non-negative, statement (i) is proved.
In view of statement (i) representation \eqref{Q1} follows (see \cite[Theorem 4.2]{Watanabe}). 
Finally, since both representations \eqref{Q1} and \eqref{Q2} hold, 
it follows that $\rmd\nu(v)=e^{-v}\varphi(e^{-v})\,\rmd v$, and
$\nu(v)$ is an integral (i.e., $\nu(v)=\int_0^v e^{-u}\varphi(e^{-u})\,\rmd u$).
The lemma is then proved.
\end{proof}

Then we can prove the following theorem.

\begin{theorem}
\label{the:N2}
Consider the power series $\frac{1}{2\pi}\sum_{k=0}^\infty g_k z^k$ ($z\in\C$, $|z|<1$), and suppose that all the
conditions required in Lemma \ref{lem:N1} are satisfied. Moreover, suppose that $\tg(0)=1$, then $\tg(\ell)$
is the Laplace transform of a probability distribution $F$, and the function $e^{-v}\varphi(e^{-v})$ is the 
density of a probability distribution.
\end{theorem}

\begin{proof}
Since $\tg(\ell)$ is completely monotonic (as proved in Lemma \ref{lem:N1}) and $\tg(0)=1$, it follows 
that $\tg(\ell)$ is the Laplace transform of a probability distribution $F$ 
(see \cite[Chapter XIII, Section 4, Theorem 1]{Feller}). Therefore, in view of the
fact that both representations \eqref{Q1} and \eqref{Q2} hold, it follows that $e^{-v}\varphi(e^{-v})$ 
is the density of a  probability distribution, and the theorem is proved.
\end{proof}

\begin{remarks}
(i) The result of Theorem \ref{the:N2} can be obtained more rapidly by assuming that the set 
$\{g_k\}_{k=0}^\infty$ of the coefficients of the power series $\frac{1}{2\pi}\sum_{k=0}g_k z^k$
($z\in\C$, $|z|<1$) satisfies the conditions of Proposition \ref{pro:1} (along with $g_0=1$) and also the conditions
of Proposition \ref{pro:3} with $p=2+\varepsilon$. In such a case, by Proposition \ref{pro:4} we can 
say that $\varphi(t)$ in \eqref{1.16} is the density function of a probability distribution concentrated on the 
closed interval $[0,1]$. Then, through the standard transformation $t=e^{-v}$, the result of 
Theorem \ref{the:N2} is rapidly obtained. 

(ii) Let us, however, point out that the situation is radically different for distributions that are concentrated 
on some finite interval with respect to those which are not. In fact, in general, a distribution is not uniquely 
determined by its moments (see \cite[Chapter VII, p. 224]{Feller}). In our case, however, the conditions required by 
Proposition \ref{pro:3} with $p=2+\varepsilon$ allow us to invoke Carlson's theorem, and, accordingly, to conclude 
that the Carlsonian interpolation of the coefficients is unique, which, in turn, guarantees the uniqueness of the 
distribution in the present case. 
\end{remarks}

\begin{theorem}
\label{the:N3}
Consider the power series $\frac{1}{2\pi}\sum_{k=0}^\infty g_k z^k$ ($z=x+\rmi y\in\C$, $x,y\in\R$, $|z|<1$),
and suppose that all the conditions required by Theorem \ref{the:N2} are satisfied.
Moreover, let us assume that the set of numbers $\{f_k\}_{k=0}^\infty$, where $f_k=(k+1)g_k$,
satisfies condition \eqref{1.15} of Proposition \ref{pro:3} with $p=2+\varepsilon$.
Then, this series admits a holomorphic extension to the cut-domain $\SSS=\{z\in\C\setminus[1,+\infty)\}$.
Moreover, the function $\oJ(x)\doteq\frac{J(x)}{x}$ ($x\in[1,+\infty)$), associated with the jump
function $J(x)$ across the cut $x\in[1,+\infty)$, is the density of a probability distribution. 
\end{theorem}

\begin{proof}
Since the numbers $\{f_k\}_{k=0}^\infty$ ($f_k=(k+1)g_k$) satisfy condition \eqref{1.15} of 
Proposition \ref{pro:3} with $p=2+\varepsilon$, it follows from
Proposition \ref{pro:5} that the coefficients $g_k$ are the restriction 
to non-negative integers of the Laplace transform  of the jump function across the cut. 
On the other hand, by Theorem \ref{the:N2}, the unique Carlsonian interpolation $\tg(\ell)$ of the set 
$\{g_k\}_{k=0}^\infty$ is the Laplace transform of a probability distribution. It follows that the jump 
function across the cut is the density of a probability distribution.
Passing from the geometry of the $\tau$-plane to that of the $z$-plane and, accordingly,
from the Laplace transform to the Mellin transform, it follows that
the function $\oJ(x)\doteq\frac{J(x)}{x}$ ($x\in[1,+\infty)$), 
where $J(x)$ denotes the jump function across the cut,
is the density of a probability distribution (see the example below).
\end{proof}

\begin{example}
Consider the sequence $\{g_k\}_{k=0}^\infty$, with $g_k=6[(k+3)(k+2)]^{-1}$. We have
\beq
\tg(\ell)=6\int_0^{+\infty} e^{-\ell v}\,e^{-v}\,(e^{-v}-e^{-2v})\,\rmd v;
\quad \tg(0)=1; \quad \tg(\ell)|_{\ell=k}=g_k.
\label{2.34}
\nonumber
\eeq
Moreover, $\varphi(t)=6(t-t^2)\geqslant 0$ for $t\in[0,1]$. 
For what regards the sequence $\{f_k\}_{k=0}^\infty$, with $f_k=6(k+1)[(k+2)(k+3)]^{-1}=(k+1)\,g_k$, we have
\beq
\tf(\ell)=6\int_0^{+\infty} e^{-\ell v}\,e^{-v}\,(2e^{-2v}-e^{-v})\,\rmd v \quad
\mathrm{and} \quad \tf(\ell)|_{\ell=k}=f_k.
\label{2.34bis}
\nonumber
\eeq
Therefore all the requirements of Theorems \ref{the:N2} and \ref{the:N3} are satisfied. We can thus state
that the function $6(e^{-2v}-e^{-3v})$ is the density of a probability distribution.
The standard transformation from the $\tau$-plane to the $z$-plane, and, accordingly,
the passage from the Laplace transform to the Mellin transform leads to the jump function
$J(x)=6(x^{-2}-x^{-3})$ for $x\in [1,+\infty)$ (and $J(x)=0$ for $x\in[0,1)$).
Finally, $\oJ(x)=6(x^{-3}-x^{-4})$ is the density of a probability distribution. 
\end{example}

\vspace*{0.1cm}

\section{Numerical approximation of the jump function across the cut}
\label{se:reconstruction}

\subsection{Preliminaries}
\label{subse:reconprelim}

As we have already remarked in the Introduction, a very relevant problem,
connected with the holomorphic extension associated with series expansions, is the
reconstruction of the discontinuity function across the cut starting
from the series coefficients. In these preliminary considerations we focus
on the power series, but the same arguments hold also for trigonometric series.
In particular, from the results of Theorem \ref{the:1} we know that the
function $G(z)$ belongs to the space $ H^2(\SSS)$ (see the Introduction
and, more specifically, Stein-Wainger's Theorem 3 \cite{Stein}, where
this result is proved in detail). Accordingly, $G(z)$ and the jump function
$J(x)=-\rmi[G_+(x)-G_-(x)]$ ($G_\pm(x)\doteq\lim_{y\to 0;y>0}G(x\pm\rmi y)$, $x\in[1,+\infty)$) are
related by the following Cauchy-type integral equation:
\beq
G(z)=\int_1^{+\infty}\frac{J(x)}{x-z}\,\rmd x
\qquad (G(z)\in H^2(\SSS)).
\label{V:1}
\nonumber
\eeq
One is then naturally led to the problem of reconstructing the jump function
$J(x)$ starting from $G(z)$, or more precisely, from the approximate knowledge
of a finite number of coefficients in the expansion of $G(z)$ within the open unit disk $|z|<1$.
This problem is obviously ill-posed in the sense of Hadamard \cite{Hadamard} since only a finite number of data 
are used in the calculations and, moreover, these data are necessarily affected by errors
(at least, roundoff numerical errors). This question is strictly connected to the numerical
analytic continuation from a region entirely contained in the unit disk up to
the unit circle. Indeed, the cut $z$-plane can be conformally mapped onto the
unit disk in the $\zeta$-plane; in this map the upper (resp., lower) lip of the cut
is mapped in the upper (resp., lower) half of the unit circle. Therefore, the continuation
up to the cut corresponds to the continuation up to the unit circle $|\zeta|=1$.

Analogous to the previous one, it is the so-called \emph{moment problem}, i.e.,
reconstructing the function $\varphi(t)$ from a finite number of moments $\{\mu_k\}_{k=0}^{N_0}$
(see Propositions \ref{pro:3} and \ref{pro:4}, and formula \eqref{1.16}).

All these problems are ill-posed in the sense of Hadamard \cite{Hadamard}, and their numerical
solution requires regularization \cite{Tikhonov}.
The standard regularizing procedure is based on the following two relevant results:
\begin{itemize}\setlength\itemsep{0.5em}
\item[(a)] \textit{A theorem on compactness}. Let $\sigma$ be a continuous map on a compact
topological space into a Hausdorff topological space; if $\sigma$ is one-to-one, then its
inverse map $\sigma^{-1}$ is continuous \cite[p. 141]{Kelley}. By this theorem the
compactness of the solution space is sufficient to guarantee the continuity
of $\sigma^{-1}$ if $\sigma$ is one-to-one (as in the case of the analytic continuation
in view of the uniqueness of the holomorphic continuation).
\item[(b)] \textit{The dependence of the regularized solution on the data accuracy}. From this
viewpoint two cases must be distinguished: (i) the reconstructed solution presents
a H\"older-type dependence on the data accuracy or noise (denoted here by $\varepsilon$),
i.e., it is of the form $\varepsilon^\gamma$ ($0<\gamma<1$);
(ii) the best possible stability estimate is at most only logarithmic, that is, proportional
to $|\log\varepsilon|^{-1}$.
\end{itemize}
For what concerns the first point it can be noted that the compactness of the solution space
is usually obtained by means of \emph{a-priori} global bounds on the solution, which are appropriate
to restrict suitably the solution space to a compact subspace, where the regularized
solution is looked for. As a typical example, remind the analytic
continuation: Denote by $D$ the closed unit disk $|z|\leqslant 1$ ($z\in\C$); then a global
\emph{a-priori} bound of the following type: $\sup_{z\in\partial D}|G(z)|\leqslant E$,
where $\partial D$ is the boundary of $D$, and $E$ is a constant independent of the
functions $G(z)$ (which are supposed to be holomorphic in $|z|<1$ and continuous on $|z|=1$)
is sufficient to restore the stability in the continuation from a region $\Gamma$ completely
contained in the unit disk, where the function $G(z)$ is supposed to be approximately known, up
to any compact domain entirely contained in $D$. Let us, however, note that the condition
above is not sufficient to perform a continuation up to the boundary of $D$.
In such a case a bound of the following form:
$\sup_{z\in\partial D} |\rmd G/\rmd z|\leqslant E'$ ($E'$ constant independent of $G$)
must also be required.

Now, we come to the point (b). It can be shown that, in the continuation up to a compact subdomain of the
unit disk, the restored stability is of H\"older-type, precisely, of the form:
$\varepsilon^{\omega(z)}E^{(1-\omega(z))}$, where $\omega(z)$ is the harmonic function on
$D\setminus\Gamma$, continuous on $\overline{D}$, which is equal to 1 on $\Gamma$, and equal to zero
on $\partial D$ \cite{Miller}. In the continuation up to the boundary the dependence on the data accuracy 
$\varepsilon$ is of logarithmic type \cite{John}, i.e., the error function is bounded by 
$C|\log(\varepsilon)|^{-1}$ ($C$ being a constant).
From this analysis it could appear reasonable to conclude that in certain cases
the restored stability presents a dependence on the data accuracy which is so poor
that the reconstruction from the data should never be attempted. This, in particular, would seem to be the
case of the analytic continuation up to the boundary of the unit disk, which is equivalent
to the reconstruction of the jump function across the cut. If such a conclusion would be
effectively true, then any attempt to reconstruct the discontinuity across the cut
should fail. However, let us observe that the results exposed in the points (a) and (b)
are sometimes taken in a rather dogmatic form, so that we think it is worth presenting
in the next subsection an alternative approach to this problem. We shall show, also
by numerical examples, that at least in a significant number of cases the reconstruction of
the jump functions, and the solution of the \emph{moment problem} as well, can be effectively performed.

First, observe that the requirement of compactness of the solution space,
indicated in the point (a), is sufficient to regularize the solution, but it is
not necessary. Secondly, we note that in practical problems
(and specifically in those suggested by physical problems) one does not know, in general,
appropriate \emph{a-priori} global bounds necessary for obtaining a compact subspace of the
solution space. It follows that the so-called standard approach to regularization
is often impracticable. Instead, the method we shall present in the next subsection
does not require any \emph{a-priori} global bound on the solution.

For what concerns the point (b), a great attention has been usually
devoted to the dependence on the data accuracy $\varepsilon$, but it has not been taken
enough heed of the class of functions that one wants to reconstruct. More precisely,
the smoothness of the functions to be reconstructed (at least
the $C^0$-continuity) is very relevant. We shall show in Section \ref{subse:numerical}
by numerical examples that, if the jump function is continuous and the
data accuracy is \emph{sufficiently good}, then the reconstruction is satisfactory, while it is
not in the case of discontinuous functions.

\vspace*{0.2cm}

\subsection{Power series: numerical approximation of the jump function across the cut}
\label{subse:jump}

In \cite{DeMicheli1,DeMicheli2} we proved the following results\footnote{In \cite{DeMicheli1,DeMicheli2} some errors 
and misprints must be corrected:
(i) In \cite{DeMicheli1,DeMicheli2} we must write: $f_n=(n+1)^p a_n$ ($p\geqslant 1$) instead of
$f_n=n^p a_n$ ($p\geqslant 1$). (ii) In \cite{DeMicheli2}, on the r.h.s.
of formula (33) one must read $\phi_m(x)$ instead of the erroneously written $\psi_m(x)$.
(iii) In formula (35) of \cite{DeMicheli2} one must read `$\leqslant$` instead of `$=$`.
These modifications do not change all the results of those papers.}.

\begin{theorem}
\label{the:4}
Consider the power series $\frac{1}{2\pi}\sum_{k=0}^\infty g_k z^k$
($z\in\C, |z|<1$), and suppose that the coefficients $\{g_k\}_{k=0}^\infty$
satisfy the conditions of Proposition \ref{pro:3} with $p=2+\varepsilon$
($\varepsilon>0$ arbitrarily small). Then, the jump function $J(x)$ across the cut $x\in[1,+\infty)$
can be represented by the following expansion that converges in the $L^2$-norm:
\beq
J(x) = \sum_{n=0}^\infty c_n\,\Phi_n(x)
\qquad (x\in[1,+\infty)),
\label{3.1}
\eeq
where the coefficients $c_n$ are given by
\beq
c_n = \sqrt{2}\sum_{k=0}^\infty\frac{(-1)^k}{k!} \,g_k\,P_n\left[-\rmi\left(k+\frac{1}{2}\right)\right],
\label{3.2}
\eeq
$P_n(\cdot)$ being the Meixner-Pollaczek polynomials \cite{Bateman,Szego}. The functions $\Phi_n(x)$ form 
an orthonormal basis in $L^2(0,+\infty)$, and are given by
\beq
\Phi_n(x) = \rmi^n\sqrt{2}\,L_n\left(\frac{2}{x}\right) \, \frac{e^{-\frac{1}{x}}}{x},
\label{3.3}
\eeq
$L_n(\cdot)$ being the Laguerre polynomials.
\end{theorem}

\begin{proof}
The proof is given in Theorems 2 and 2' of \cite{DeMicheli1} (see also \cite{DeMicheli2}).
\end{proof}

\noindent
The polynomials $P_n(\cdot)$, introduced in \eqref{3.2}, are a
particular case of the Pollaczek polynomials $P_n^{(\alpha)}(y)$,
orthonormal with respect to the following weight function:
$w(y) = \frac{1}{\pi}2^{(2\alpha-1)}|\Gamma(\alpha+\rmi y)|^2$, ($\alpha>0, y\in\R$),
i.e., 
\beq
\int_{-\infty}^{+\infty}w(y)P_n^{(\alpha)}(y)P_m^{(\alpha)}(y)\,\rmd y =
\delta_{n,m}\frac{\Gamma(n+2\alpha)}{\Gamma(n+1)\Gamma(2\alpha)}. \nonumber
\label{add6}
\eeq
Here we take $\alpha=\frac{1}{2}$, which, for simplicity, is hereafter omitted in the notation.

\vspace{0.5cm}

Obviously, in actual numerical computation use cannot be made of the whole sequence $\{g_k\}_{k=0}^\infty$, 
but the input data are given by only a finite number of coefficients $\{g_k\}_{k=0}^{N}$ which, 
in addition, are necessarily affected by noise. 
The noisy coefficients will be denoted by $g_k^{(\varepsilon)}$, and we assume that the noise is such that
$|g_k-g_k^{(\varepsilon)}|\leqslant\varepsilon$ ($k=0,1,2,\ldots, N$; $\varepsilon>0$). 
Next, we introduce the following approximate coefficients, which can be actually computed 
from the input data $\{g_k^{(\varepsilon)}\}_{k=0}^{N}$:
\beq
c_n^{\,(\varepsilon,N)} \doteq \sqrt{2}\sum_{k=0}^N\frac{(-1)^k}{k!}
\,g_k^{(\varepsilon)}\,P_n\left[-\rmi\left(k+\frac{1}{2}\right)\right].
\label{3.4}
\nonumber
\eeq
Evidently, we have: $c_n^{\,(0,\infty)}=c_n$. We can now state the following lemma.

\begin{lemma}
\label{lem:6}
Consider the power series $\frac{1}{2\pi}\sum_{k=0}^\infty g_k z^k$ ($z\in\C, |z|<1$),
and suppose that the coefficients $\{g_k\}_{k=0}^\infty$ satisfy the conditions of
Proposition \ref{pro:3} with $p=2+\varepsilon$ ($\varepsilon>0$ arbitrarily small). Then
the following statements hold:
\begin{flalign*}
(\mathrm{i}) &&&
\sum_{n=0}^\infty \left|c_n^{\,(0,\infty)}\right|^2 = \left\|J\right\|^2_{L^2(1,+\infty)}
=K \qquad (K = ~\mathrm{constant}).
\qquad\qquad\qquad\qquad\qquad\qquad \\
(\mathrm{ii}) &&&
\sum_{n=0}^\infty \left|c_n^{\,(\varepsilon,N)}\right|^2 = +\infty. \\
(\mathrm{iii}) &&&
\lim_{\staccrel{N\to\infty}{\scriptstyle\varepsilon\to 0}} c_n^{\,(\varepsilon,N)}
=c_n^{\,(0,\infty)}=c_n \qquad (n=0,1,2,\ldots), \mathit{the~ order~ of~ the~ limits} \\
   &&&
\mathit{being~ relevant}\!:
\mathit{first} \lim_{N\to\infty}, \mathit{then} \lim_{\varepsilon\to 0}.
\end{flalign*}
$\mathrm{(iv)}$ \ If $m_0(\varepsilon,N)$ is defined as
\beq
m_0(\varepsilon,N)\doteq\max\left\{m\in\N\,:\,\sum_{n=0}^m\left|c_n^{\,(\varepsilon,N)}\right|^2
\leqslant K\right\},
\label{3.8}
\eeq
\indent\indent then
\beq
\lim_{\staccrel{N\to\infty}{\scriptstyle\varepsilon\to 0}}
m_0(\varepsilon,N) = +\infty.
\label{3.9}
\eeq
$\mathrm{(v)}$ The sum
\beq
M_m^{\,(\varepsilon,N)}\doteq\sum_{n=0}^m\left|c_n^{\,(\varepsilon,N)}\right|^2 \qquad (m\in\N)
\label{3.10}
\eeq
\indent satisfies the following properties:
\begin{itemize}
\item[(a)] it does not decrease for increasing values of $m$;
\item[(b)] the following relationship holds:
\beq
M_m^{\,(\varepsilon,N)}\geqslant\left|c_m^{\,(\varepsilon,N)}\right|^2
\staccrel{\sim}{\scriptstyle m\to\infty} A^{(\varepsilon,N)} \cdot (2m)^{2N} \qquad (N~\mathrm{fixed}),
\label{3.11}
\eeq
where $A^{(\varepsilon,N)}=2\left|g_N^{(\varepsilon)}\right|^2/(N!)^4$ is a quantity independent of $m$.
\end{itemize}
\end{lemma}

\begin{proof}
See Lemma 1 of \cite{DeMicheli2}. 
\end{proof}

\begin{remark}
From statement (v) and formula \eqref{3.9} it follows that, for large values of $N$
and small values of $\varepsilon$, the sum $M_m^{\,(\varepsilon,N)}$
exhibits (as a function of $m$) a \emph{plateau}, i.e., a range of $m$-values
where it is nearly constant. The upper limit of this range is given by $m_0(\varepsilon,N)$ 
(see the next subsection for numerical examples).
\label{rem:4}
\end{remark}

\vskip 0.2cm

\noindent
We now introduce the following approximation of the jump function across the cut:
\beq
J^{(\varepsilon,N)}(x) \doteq \sum_{n=0}^{m_0(\varepsilon,N)} c_n^{\,(\varepsilon,N)}\,\Phi_n(x)
\qquad (x\in[1,+\infty)),
\label{3.12}
\eeq
and state the following theorem.

\begin{theorem}
\label{the:5}
Consider the power series $\frac{1}{2\pi}\sum_{k=0}^\infty g_k z^k$ ($z\in\C, |z|<1$),
and suppose that the coefficients $\{g_k\}_{k=0}^\infty$ satisfy the conditions of
Proposition \ref{pro:3} with $p=2+\varepsilon$ ($\varepsilon>0$ arbitrarily small). Then
the following equality holds:
\beq
\lim_{\staccrel{N\to\infty}{\scriptstyle\varepsilon\to 0}}
\left\|J-J^{(\varepsilon,N)}\right\|_{L^2(1,+\infty)} = 0.
\label{3.13}
\eeq
\end{theorem}

\begin{proof}
See Theorem 3 of \cite{DeMicheli2}.
\end{proof}

\vskip 0.3cm

\begin{remark}
In Theorem \ref{the:4} (see formula \eqref{3.1}) and in Theorem \ref{the:5} (formulae \eqref{3.12} 
and \eqref{3.13}), we present the reconstruction of the jump function on the semiaxis $x\in[1,+\infty)$. 
But we could consider as well the reconstruction of the discontinuity function on the semiaxis 
$x\in[0,+\infty)$, noting that the functions $\Phi_n(x)$ form a basis in $L^2(0,+\infty)$ 
(see \eqref{3.3} in Theorem \ref{the:4}). We should expect that in the interval $x\in[0,1)$ 
the reconstruction of the jump function converges to zero in the sense of the $L^2$-norm, as we actually 
observe, e.g., in Figures \ref{fig:3} and \ref{fig:4}.
\label{rem:last}
\end{remark}

\vspace*{0.2cm}

At this point we notice that on one hand we obtain in Theorem \ref{the:4} (see expansion \eqref{3.1})
the reconstruction of the jump function across the cut in the ideal case of an infinite number
of noiseless input data. Expansion \eqref{3.1} converges in the $L^2$-norm. On the other hand, 
if we realistically assume that the data are perturbed by noise and finite in number, we can prove
Theorem \ref{the:5} without the use of any \emph{a-priori} global bound. Further,
by using Schwarz's inequality, statements (iv) and (v) of Lemma \ref{lem:6}, and equality \eqref{3.12}
(see also Remark \ref{rem:last}) we can also state that the error 
\beq
\Delta J^{(\varepsilon,N)}\doteq J(x)-J^{(\varepsilon,N)}(x) \nonumber
\label{add7}
\eeq
is bounded as follows:
$\|\Delta J^{(\varepsilon,N)}\|_{L^2(0,+\infty)}\leqslant 2\|J\|_{L^2(0,+\infty)}$. The last inequality 
guarantees that the root mean square error is bounded; furthermore, Theorem \ref{the:5}
(see also Remark \ref{rem:last}) shows that it tends to zero as $N\to\infty$ and $\varepsilon\to 0$.
Unfortunately, this theorem does not indicate how fast $\|\Delta J^{(\varepsilon,N)}\|_{L^2(0,+\infty)}$
runs to zero. Anyway, it is worth recalling that the present problem is very close
to the problem of the analytic continuation from a region entirely contained in the unit disk
up to the unit circle (see Section \ref{subse:reconprelim}). Then we can reasonably argue
that the stability estimate which holds in that case remains true in the present one.
Now, a classical result due to F. John (see \cite{John}) states that in the analytic continuation
up to the boundary the dependence on the data accuracy is of logarithmic type (see again Section \ref{subse:reconprelim}),
as it is typical of severely ill-posed problems.
Translating this result to our case, we can argue that 
$\|\Delta J^{(\varepsilon,N)}\|_{L^2(0,+\infty)}$ tends to zero, for $\varepsilon$ tending to zero,
only logarithmically, i.e., as $|\log\varepsilon|^{-1}$. This behavior has been also shown 
by numerically examples for this type of problem (namely, the approximation of the jump 
function across the cut) in \cite{DeMicheli2} (see, in particular, Figure 4C).
In the present case we must also take into account the dependence on $N$. Now, we conjecture
that the stability estimate in our problem is of the following type: for small $\varepsilon$ and
large $N$, 
$\|\Delta J^{(\varepsilon,N)}\|_{L^2(0,+\infty)}\leqslant C \left(|\log\varepsilon|\log N\right)^{-1}
\longrightarrow 0$ for $N\to\infty$ and $\varepsilon\to 0$ ($C= \mathrm{constant}$).
Let us note that the last conjecture is enforced by several numerical examples 
performed in \cite{DeMicheli2}, which the interested reader is refereed to.

We see that if the bound illustrated above is correct, the decrease of the root mean square error
for $N\to\infty$ and $\varepsilon\to 0$ is slow, and this behaviour is a consequence of the fact that 
the current problem is severely ill-posed. Therefore we can conjecture that the numerical
approximation of the jump function across the cut (starting from a finite number of noisy 
coefficients $\{g_k^{(\varepsilon)}\}_{k=0}^N$) can be satisfactory only if the 
following conditions are satisfied:
\begin{itemize}\setlength\itemsep{0.4em}
\item[(a)] The jump function must be at least of class $C^0$ and of bounded variation.
\item[(b)] The bound $\varepsilon$ on the noise must be very small.
\item[(c)] The number $(N+1)$ of data must be large.  
\end{itemize}
We shall see in the numerical examples that in order to obtain a satisfactory approximation 
all these three conditions should be fulfilled.

\vskip 0.4cm

\subsection{Numerical examples}
\label{subse:numerical}

Consider the power series $\frac{1}{2\pi}\sum_{k=0}^\infty g_k z^k$
($z\in\C, |z|<1$) with coefficients $g_k=\frac{6}{(k+2)(k+3)}$ ($k=0,1,2,\ldots$). 
They are the restriction to the integers of the following Laplace transform:
$\tg(\ell)=6\int_{0}^{+\infty}e^{-\ell v}e^{-v}(e^{-v}-e^{-2v})\,\rmd v$;
$\tg(\ell)|_{\ell=k}=g_k$ for $k=0,1,2,\ldots$.
Moreover, the following equality holds: $\tg(0)=1$.
The corresponding jump function is then:
$J(v)=6(e^{-2v}-e^{-3v})$, and $J(0)=0$.
Moving from the $2\pi$-periodic $\tau$-plane
($\tau=u+\rmi v$; $u,v\in\R$) to the complex $z$-plane by means of the mapping $z=e^{-\rmi\tau}$, 
we obtain the jump function
$J(x)=6(x^{-2}-x^{-3})$ for $x\in[1,+\infty)$, and $J(x)=0$ for $x\in[0,1)$. Notice that $J(1)=0$, so that,
in $x=1$, $J(x)$ is continuous but not differentiable.
By using the numerical procedure outlined in the previous subsection,
the jump function $J(x)$ can be approximated starting from a finite
number $(N+1)$ of coefficients, and the results are summarized in Figure \ref{fig:3}.
In Figure \ref{fig:3}a the plot of the sum $M_m^{\,(\varepsilon,N)}$ (see \eqref{3.10}) against $m$ is given
for various values of $N$. In this case $\varepsilon=0$, which means that no noise has been added
to the input coefficients $g_k$, and that the only source of
inevitable noise is given by the numerical roundoff error. As stated in Remark \ref{rem:4},
the sum $M_m^{\,(\varepsilon,N)}$ manifests a plateau, which in this example begins at $m \simeq 10$ and whose length
varies according to the number $N$ of input coefficients. Moreover, it can be seen that, after the end of the plateau,
$M_m^{\,(\varepsilon,N)}$ starts to diverge as a power of $m$ (see \eqref{3.11}) for the presence
of the roundoff noise and for the finiteness of the number of input coefficients.
Varying $N$ from $N=20$ to $N=60$ the length of the plateau increases correspondingly,
reflecting the increase of information available for the reconstruction. For $N \gtrsim 60$
the range of the plateau does not vary appreciably, being limited superiorly by the effect
of the roundoff noise. 
Actually, the sum in the approximation \eqref{3.12} can be truncated at a value 
$m_\mathrm{t}(\varepsilon,N)\leqslant m_0(\varepsilon,N)$, provided it belongs to the plateau.
Indeed, from definition \eqref{3.8} and in view of Remark \ref{rem:4}, it results that
the choice of the truncation number $m_\mathrm{t}(\varepsilon,N)$ (provided it lies within the plateau)
is not critical for the accuracy of the final result.
Figure \ref{fig:3}b shows the approximation $J^{(0,N)}(x)$ of the jump function (solid line),
computed by means of formula \eqref{3.12} with $N=60$ noiseless input coefficients $g_k$, 
and $m_\mathrm{t}(0,60)=100$.
It is evident the high quality of the reconstruction, the latter being almost
indistinguishable from the true jump function $J(x)$ (dashed line) for the most part of the $x$-range.
This excellent reconstruction has been obtained also with a smaller number ($N\gtrsim 20$) of input coefficients
(not shown, for brevity).
It is worth observing that for $x\in[0,1]$, where the true jump function is null, the reconstruction
oscillates evidently as a consequence of the $L^2$-character of the approximation of $J^{(\varepsilon,N)}(x)$
to $J(x)$ (see Theorems \ref{the:4} and \ref{the:5}).

In Figure \ref{fig:3}c the sum $M_m^{\,(\varepsilon,N)}$ is plotted against $m$ for various values of $\varepsilon$
(see the figure legend for numerical details). Similarly to what shown in panel (a), it can be seen that the length
of the plateau increases as the noise level decreases (for $\varepsilon=10^{-7}$ the plateau ranges from $m \simeq 10$
through $m \simeq 40$). An example of reconstruction of the jump function in the case of noisy input coefficients
($N=60$) is given in panel (d) for $\varepsilon =10^{-6}$; the truncation number we used is
$m_\mathrm{t}(10^{-6},60)=26$, which lies within the plateau visible in panel (c). Even in this case the reconstruction 
of the jump function across the cut is quite satisfactory.

\begin{figure}[tb]
\begin{center}
\leavevmode
\includegraphics[width=10.6cm]{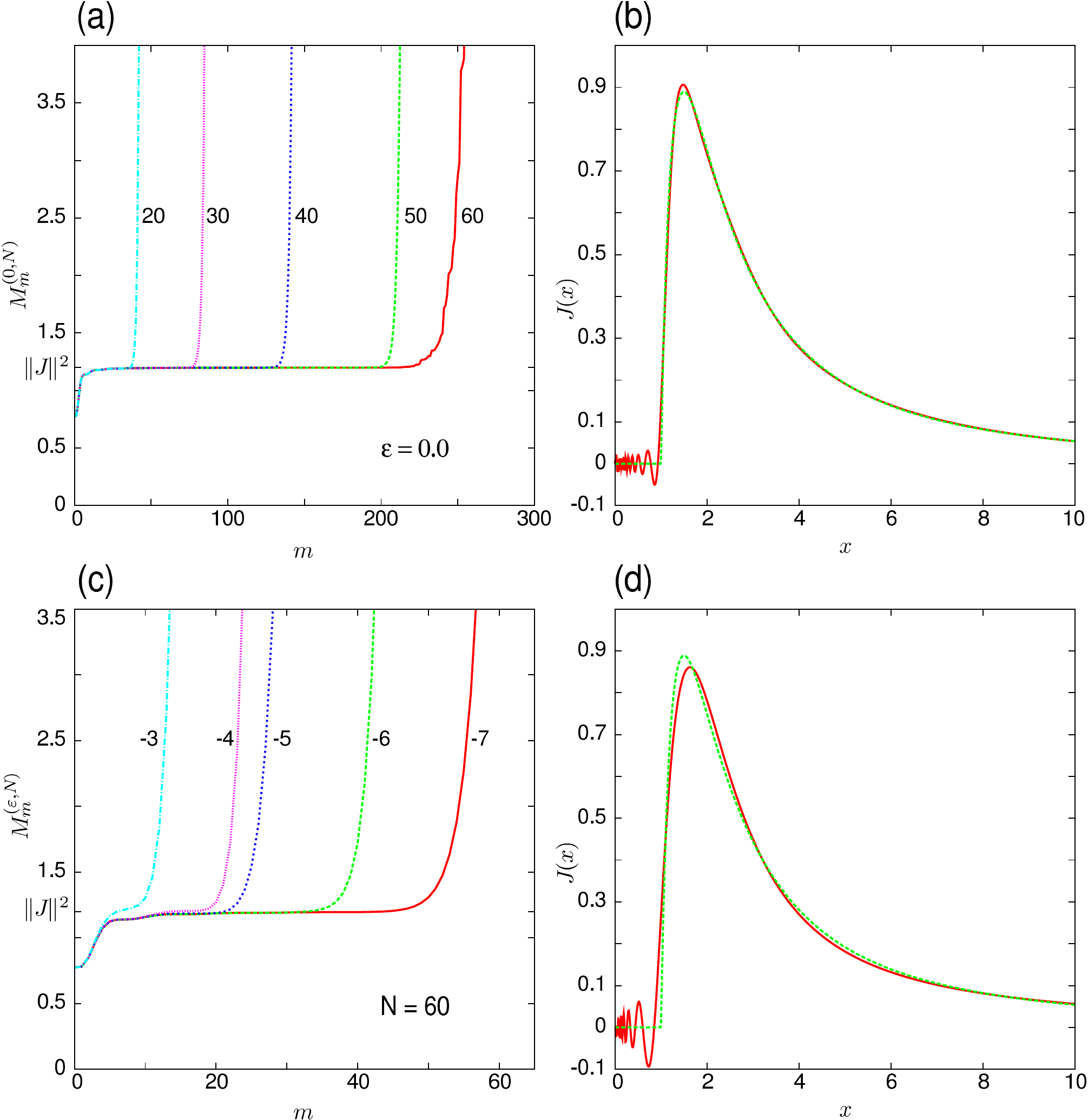}
\caption{\label{fig:3}
\small Numerical approximation of the jump function across the cut. Coefficients $g_k =6[(k+2)(k+3)]^{-1}$.
(a) Plot of the sum $M_m^{(0,N)}$ vs. $m$ (see \eqref{3.10}) at various values of the number $N$ of input noiseless
coefficients: $N=20,30,40,50,60$.
(b) Approximation (solid line) of the jump function, computed by using equation \eqref{3.12} with $N=60$, 
and truncation number $m_\mathrm{t}(0,60)=100$.
The true jump function (dashed line) is: $J(x)= 6(x^{-2}-x^{-3})$ for $x\in[1,+\infty)$ and null elsewhere.
(c) Plot of the sum $M_m^{(\varepsilon,60)}$ vs. $m$ (see \eqref{3.10}) at various values of noise level:
$\varepsilon=10^{-3},10^{-4},10^{-5},10^{-6},10^{-7}$.
The noisy coefficients $g_k^{(\varepsilon)}$ are computed by adding to the noiseless coefficient $g_k$ 
a random variable uniformly distributed in the interval $[-\varepsilon,\varepsilon]$.
(d) Approximation (solid line) of the jump function, computed by using equation \eqref{3.12} 
with $N=60$, $\varepsilon=10^{-6}$, $m_\mathrm{t}(10^{-6},60)=25$.}
\end{center}
\end{figure}

Analogously, starting from another finite sequence of coefficients $g_k=\frac{1}{k+1}$, we can reconstruct the
jump function, which in this case is: $J(x)=x^{-1}$ for $x\in[1,+\infty)$, and $J(x) = 0$ for $x\in[0,1)$. 
Notice that $J(1)=1$, and
therefore, in this case, $J(x)$ is discontinuous in $x=1$.
The plot of the sum $M_m^{\,(\varepsilon,N)}$, given in Figure \ref{fig:4}a, shows that even when $\varepsilon=0$ 
a range of $m$-values where it is nearly constant is not neatly displayed.
As expected, the approximation of the jump function, which is plotted in Figure \ref{fig:4}b, is not very satisfactory,
in agreement with the remarks illustrated in Section \ref{subse:reconprelim}.

\begin{figure}[tb]
\begin{center}
\leavevmode
\includegraphics[width=10.5cm]{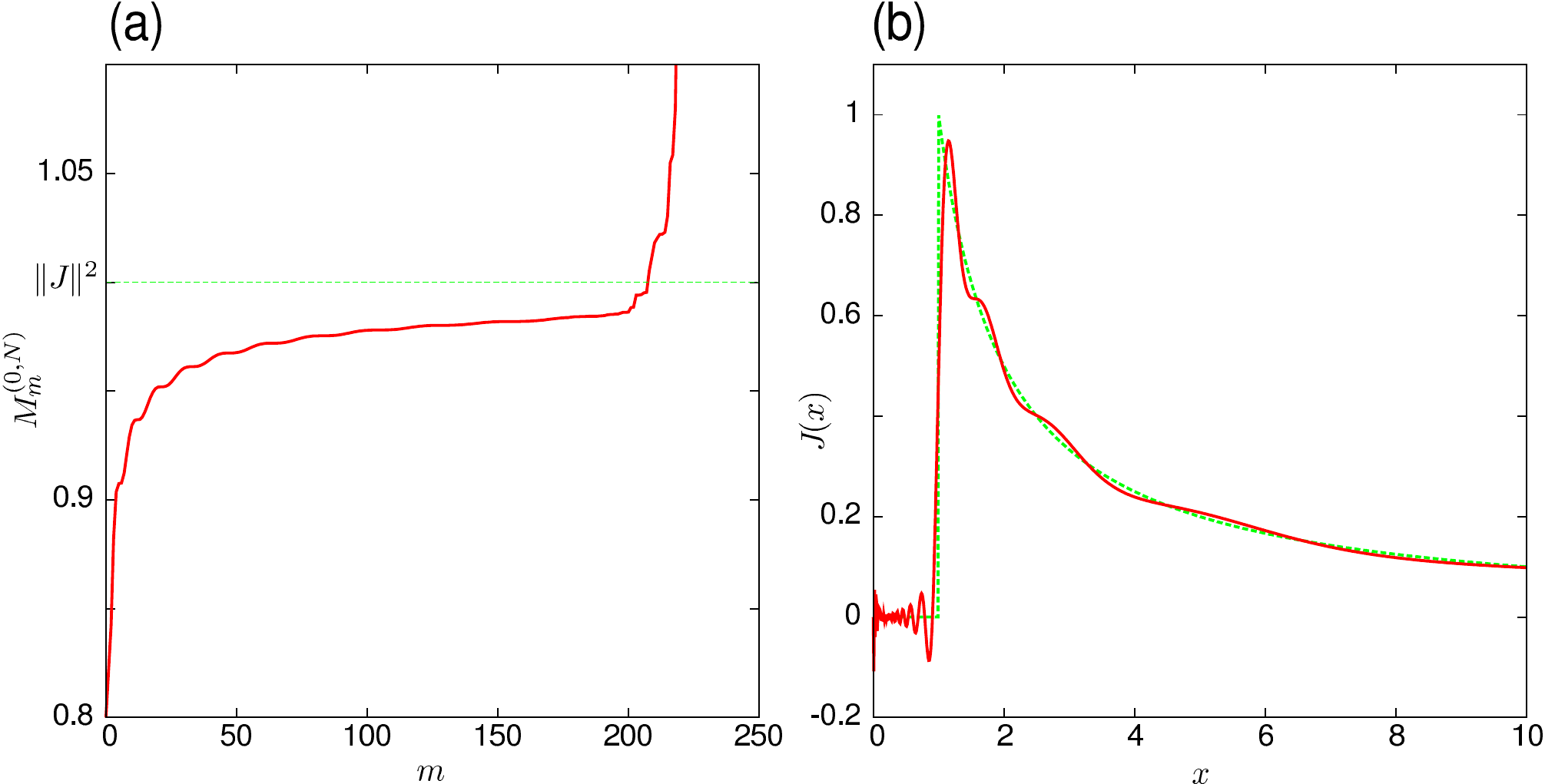} 
\caption{\label{fig:4}
\small Numerical approximation of the jump function across the cut. Coefficients $g_k =(k+1)^{-1}$.
(a) Plot of the sum $M_m^{(0,N)}$ vs. $m$ (see \eqref{3.10}) with $N=60$ noiseless coefficients.
(b) Approximation (solid line) of the jump function, computed by using equation \eqref{3.12} with $N=60$, 
and truncation number $m_\mathrm{t}(0,60)=190$.
The true jump function (dashed line) is: $J(x)= x^{-1}$ for $x\in[1,+\infty)$ and null elsewhere.}
\end{center}
\end{figure}

\vskip 0.4cm

\subsection{Numerical approximation of the thermal Green functions at real time 
from those at imaginary time}
\label{subse:real}

In this subsection we consider only a system of bosons (see Subsection \ref{subse:derivation});
the case of fermions can be treated similarly.
The reconstruction procedure we are going to outline can be split in three parts along lines similar
to those followed in Section \ref{subse:jump}: (1) First, we reconstruct the function
$\tg^{(+;\rmb)}(\ell+\rmi\nu)$ ($\nu\in\R$) on the line $\ell=\frac{1}{2}$,
assuming that the Fourier coefficients $\{g_k^{(+;\rmb)}\}_{k=1}^\infty$
are known and noiseless. (2) From the knowledge of $\tg^{(+;\rmb)}(\frac{1}{2}+\rmi\nu)$
we can then reconstruct $J^{(+;\rmb)}(v)$ ($v\in\R^+$). (3) The whole procedure
is then reconsidered taking into account that the data set (i.e., the set of Fourier coefficients
which can be actually used in the computation) has a finite cardinality, and the coefficients are also 
necessarily affected by noise.

We now suppose that the Fourier coefficients $\{g_k^{(+;\rmb)}\}_{k=1}^\infty$
of the series \eqref{4.4a} satisfy the conditions of Theorem \ref{the:6}; then we can state the following
results, which have been essentially proved in \cite{Cuniberti}.

\begin{lemma}
\label{lem:7}
The function $\tg^{(+;\rmb)}(\frac{1}{2}+\rmi\nu)$ ($\nu\in\R$)
can be represented by the following series that converges in the $L^2$-norm:
\beq
\tg^{(+;\rmb)}\left(\frac{1}{2}+\rmi\nu\right)
=\sum_{n=0}^\infty d_n\,\psi_n(\nu), \nonumber
\label{4.25}
\eeq
where $\psi_n(\nu)$ denotes the Meixner-Pollaczek functions
\beq
\psi_n(\nu) \doteq \frac{1}{\sqrt{\pi}}\,\Gamma\left(\frac{1}{2}+\rmi\nu\right)\,P_n(\nu), \nonumber
\label{4.26}
\eeq
$\Gamma(\cdot)$ being the Euler gamma function, and $P_n(\nu)$ are the
Meixner-Pollaczek polynomials (introduced in Section \ref{subse:jump}). The functions $\psi_n(\nu)$ 
form an orthonormal basis in $L^2(-\infty,+\infty)$. The coefficients $d_n$ are given by
\beq
d_n=2\sqrt{\pi}\sum_{k=0}^\infty\frac{(-1)^k}{k!}g_{k+1}^{(+;\rmb)}
P_n\left[-\rmi\left(k+\frac{1}{2}\right)\right]. \nonumber
\label{4.27}
\eeq
\end{lemma}

\begin{proof}
The proof is analogous, up to slight modifications, to that of Lemma 1 in \cite{Cuniberti}.
\end{proof}

In view of statement (5) of Theorem \ref{the:6},
$e^{-v/2}J^{(+;\rmb)}(v)$ can be expressed as the inverse Fourier transform of
the function $\tg^{(+;\rmb)}\left(\frac{1}{2}+\rmi\nu\right)\in L^1(-\infty,+\infty)$
(see \eqref{4.7.bis}). Next, we introduce the following weighted $L_w^2$-space,
whose norm is defined by
\beq
\left\|f\right\|_{L_w^2(0,+\infty)} \doteq
\left(\int_0^{+\infty}w(v) \left|f(v)\right|^2\,\rmd v\right)^{\!\!\ud}
\qquad \mathrm{with} \quad w(v)\doteq e^{-v}. \nonumber
\label{4.28}
\eeq
We can then state the following theorem.

\begin{theorem}
\label{the:11}
The jump function $J^{(+;\rmb)}(v)$ ($v\in\R^+$) can be represented by the following expansion:
\beq
J^{(+;\rmb)}(v) = e^{\frac{v}{2}}\sum_{n=0}^\infty \oc_n\,\Psi_n(v)  \qquad (v\in\R^+),
\label{4.29}
\eeq
where
\beq
\oc_n=\sqrt{2}\sum_{k=0}^\infty\frac{(-1)^k}{k!}\, g_{k+1}^{(+;\rmb)}
\, P_n\left[-\rmi\left(k+\frac{1}{2}\right)\right], \nonumber
\label{4.30}
\eeq
$P_n(\cdot)$ being the Meixner-Pollaczek polynomials, and the functions $\Psi_n(v)$ are
\beq
\Psi_n(v) = \rmi^n\sqrt{2}\ L_n(2e^{-v})\,e^{-e^{-v}}\,e^{-\frac{v}{2}}, \nonumber
\label{4.31}
\eeq
$L_n(\cdot)$ being the Laguerre polynomials. The functions $\Psi_n(v)$ form an
orthonormal basis in $L^2(-\infty,+\infty)$. The convergence of the expansion
\eqref{4.29} is in the $L^2_w(0,+\infty)$-norm with weight function $w(v)=e^{-v}$.
\end{theorem}

\begin{proof}
The proof is analogous, up to minor modifications, to that of Lemma 2 and Theorem 2 of \cite{Cuniberti}.
\end{proof}

We now suppose that only a finite number of Fourier coefficients are known with a certain
degree of approximation.
Then we denote by $g_k^{(+;\rmb;\varepsilon)}$ the Fourier coefficients $g_k^{(+;\rmb)}$ perturbed by noise.
We assume that only $(N+1)$ Fourier coefficients are known within a certain approximation error
of order $\varepsilon$, i.e., the noise is assumed to be such that 
$|g_k^{(+;\rmb;\varepsilon)}-g_k^{(+;\rmb)}|\leqslant\varepsilon$
($\varepsilon>0; k=0,1,2,\ldots,N$). We introduce the finite sum
\beq
\oc_n^{\,(\varepsilon,N)}\doteq\sqrt{2}\sum_{k=0}^N\frac{(-1)^k}{k!}g_{k+1}^{(+;\rmb;\varepsilon)}
P_n\left[-\rmi\left(k+\frac{1}{2}\right)\right]. \nonumber
\label{4.32}
\eeq
Accordingly, for every $n\in\N$ we have: $\oc_n^{\,(0,\infty)}=\oc_n$. Next, we define the approximation
\beq
J^{(+;\rmb)}_{(\varepsilon,N)}(v) \doteq e^{\frac{v}{2}}
\sum_{n=0}^{\om_0(\varepsilon,N)} \oc_n^{\,(\varepsilon,N)}\,\Psi_n(v) \qquad (v\in\R^+), \nonumber
\label{4.33}
\eeq
where $\om_0(\varepsilon,N)$ is defined by
\beq
\om_0(\varepsilon,N)=\max\left\{m\in\N\,:\,\sum_{n=0}^m\left|\oc_n^{\,(\varepsilon,N)}\right|^2\leqslant \oC\right\},
\nonumber
\label{4.34}
\eeq
where the constant $\oC$ is given by
\beq
\oC = \sum_{n=0}^\infty\left|\oc_n^{\,(0,\infty)}\right|^2 = \left\|e^{-v/2}J^{(+;\rmb)}(v)\right\|^2_{L^2(0,+\infty)}.
\nonumber
\label{4.35}
\eeq
We can then state the following result.

\begin{theorem}
\label{the:12}
The following equality holds:
\beq
\lim_{\staccrel{N\to+\infty}{\scriptstyle\varepsilon\to 0}}
\left\|J^{(+;\rmb)}(v) - e^{\frac{v}{2}} \sum_{n=0}^{\om_0(\varepsilon,N)} \oc_n^{\,(\varepsilon,N)}\,\Psi_n(v) 
\right\|_{L^2_w(0,+\infty)}=0. \nonumber
\label{4.36}
\eeq
\end{theorem}

\begin{proof}
The statement follows easily by adapting the proofs of Lemma 3, Theorem 5, and Corollary 1 of 
\cite{Cuniberti}\footnote{In formula (116) of \cite{Cuniberti} one must read `$\leqslant$' 
instead of `$=$'.}.
\end{proof}

For what concerns the stability estimates in the numerical approximation of the thermal
Green function at real time we can repeat, up to some modifications essentially due to
the fact that we are treating the problem in a weighted $L^2_w$-space, the same considerations
we developed at the end of Subsection \ref{subse:jump}.

\vspace{0.5 cm}

\setcounter{equation}{0}
\renewcommand{\theequation}{A.\arabic{equation}}
\setcounter{section}{0}
\renewcommand{\thesection}{A}

\section*{Appendix: The KMS analytic structure}
\label{appendix:I}

Consider the algebra $\cA$ generated by the observables of a quantum
system. Denoting by $A,B,\ldots$ arbitrary elements of $\cA$ and by
$A\to A(t)$ ($A=A(0)$) the action of the (time-evolution) group of automorphisms on this
algebra, we study the analytic structure of two-point correlation functions
$\langle A(t_1) B(t_2) \rangle_{\Omega_\beta}$, in a thermal equilibrium state $\Omega_\beta$ of the
system at temperature $T=\beta^{-1}$.
By time-evolution invariance, these quantities only depend on $t=t_1-t_2$, and we shall put
\begin{align*}
\cW_{AB}(t) &= \langle A(t) B\rangle_{\Omega_\beta}, \\
\cW^{\,'}_{AB}(t) &= \langle B A(t) \rangle_{\Omega_\beta}.
\end{align*}
In finite volume approximations, the time-evolution is represented by a unitary group $e^{\rmi H'\,t}$,
so that
\beq
A(t)=e^{\rmi H'\,t}\,A\,e^{-\rmi H'\,t}, \nonumber
\label{A.6}
\eeq
where $H' = H-\mu N$, $H$ being the Hamiltonian, $\mu$ the chemical potential,
and $N$ the particle number; under general conditions, the operators
$e^{-\beta H'}$ have finite traces for all $\beta>0$ (see \cite{Haag}). Then the
correlation functions are given, correspondingly, by the formulae
\begin{align*}
\cW_{AB}(t) &= \frac{1}{\cZ_\beta}\mathrm{Tr}\left\{e^{-\beta H'} A(t) B\right\}, \\ 
\cW^{\,'}_{AB}(t) &= \frac{1}{\cZ_\beta}\mathrm{Tr}\left\{e^{-\beta H'} B A(t) \right\}, 
\end{align*}
where $\cZ_\beta = \mathrm{Tr}\,e^{-\beta H'}$.
One then introduces the following holomorphic functions of the complex
time variable $t+\rmi\gamma$:
\begin{subequations}
\label{A.8}
\begin{align}
& \mathrm{G}_{AB}(t+\rmi\gamma) = \frac{1}{\cZ_\beta}\mathrm{Tr}\left\{e^{-(\beta+\gamma) H'}A(t)e^{\gamma H'}B \right\},
\quad \mathrm{analytic ~in~}\{t\in\R, -\beta<\gamma<0\}, \label{A.8a} \\
& \mathrm{G}'_{AB}(t+\rmi\gamma)=
\frac{1}{\cZ_\beta}\mathrm{Tr}\left\{e^{-(\beta-\gamma) H'}B e^{-\gamma H'}A(t) \right\},
\quad \mathrm{analytic ~in~}\{t\in\R, 0<\gamma<\beta\}, \label{A.8b}
\end{align}
\end{subequations}
which are such that
\begin{align*}
& \lim_{\staccrel{\gamma\to 0}{\scriptstyle\gamma<0}} \mathrm{G}_{AB}(t+\rmi\gamma) = \cW_{AB}(t), \\
& \lim_{\staccrel{\gamma\to 0}{\scriptstyle\gamma>0}} \mathrm{G}'_{AB}(t+\rmi\gamma) = \cW^{\,'}_{AB}(t).
\end{align*}
From \eqref{A.8} and by the cyclic property of $\mathrm{Tr}$, we then obtain
the KMS relation\footnote{In \cite{Cuniberti} a factor $\cZ^{-1}_\beta$
has been erroneously omitted in all the equalities (10).}
\beq
\cW_{AB}(t)=\frac{1}{\cZ_\beta}\mathrm{Tr}\left[e^{-\beta H'} A(t) B\right]
=\frac{1}{\cZ_\beta}\mathrm{Tr}\left[B e^{-\beta H'} A(t) \right]=
\frac{1}{\cZ_\beta}\mathrm{G}'_{AB}(t+\rmi\beta), \nonumber
\label{A.10}
\eeq
which implies the identity of holomorphic functions (in the strip $0<\gamma<\beta$)
\beq
\mathrm{G}_{AB}(t+\rmi(\gamma-\beta))=\mathrm{G}'_{AB}(t+\rmi\gamma).
\label{A.11}
\eeq
According to the analysis of \cite{Haag}, in the Quantum Mechanical framework, and of
\cite{Buchholz} in the Field Theoretical framework, this KMS analytic structure is preserved in the
thermodynamic limit under rather general conditions.

In the case when the algebra $\cA$ is generated by smeared-out bosonic or fermionic field
operators (field theory at finite temperature), the principle of relativistic
causality of the theory implies additional relations for the corresponding pairs
of analytic functions $(\mathrm{G},\mathrm{G}')$. In fact, this principle of relativistic
causality is expressed by the commutativity (resp., anticommutativity) relations for
the boson field $\bPhi(x)$ (resp., fermion field $\bPsi(x)$) at space-like separation,
i.e., for $(t-t')^2<(\bx-\bx')^2$:
\beq
\left[\bPhi(t,\bx),\bPhi(t',\bx')\right]=0
\qquad\left(\mathrm{resp.},\left\{\bPsi(t,\bx),\bPsi(t',\bx')\right\}=0\right).
\label{A.12}
\eeq
In view of the properties \eqref{A.11} and \eqref{A.12} it is possible to prove
\cite{Cuniberti,Haag} that this thermal two-point function (of the field $\bPhi$ or $\bPsi$) can be
fully characterized in terms of an analytic function $\cG(t+\rmi\gamma,\bx)$
(with regular dependence in the space variables) enjoying the following properties:
\begin{itemize}
\item[(a)] $\cG(t+\rmi\gamma,\bx)=\varepsilon\,\cG(t+\rmi(\gamma-\beta),\bx)$, where
$\varepsilon=+$ for a boson field, and $\varepsilon=-$ for a fermion field;
\item[(b)] for each $\bx$, the domain of $\cG$ in the complex variable $t$ is
$\C\setminus\{t+\rmi\gamma\,:\,|t|>|\bx|,\gamma=k\beta,k\in\Z,\beta=T^{-1}\}$;
\item[(c)] the boundary values of $\cG$ at real times are the thermal correlations of the
field, namely:
\begin{align*}
& \lim_{\staccrel{\gamma\to 0}{\scriptstyle\gamma<0}} \cG(t+\rmi\gamma,\bx)=\cW(t,\bx), \\
& \lim_{\staccrel{\gamma\to 0}{\scriptstyle\gamma>0}} \cG(t+\rmi\gamma,\bx)=\cW^{\,'}(t,\bx). 
\end{align*}
\end{itemize}
In this analytic structure two quantities play an important role:
\begin{itemize}
\item[(i)] the restriction $\cG(\rmi\gamma,\bx)$ of the function $\cG$
to the imaginary axis is a $\beta$-periodic (or antiperiodic) function of $\gamma$,
which must be identified with the ``time-ordered product at
imaginary times'', considered in Matsubara's approach of imaginary time formalism \cite{Cuniberti,Haag}.
\item[(ii)] The \emph{jumps} of the function $\cG$ across the real cuts $\{t\,:\,t\geqslant|\bx|\}$
and $\{t\,:\,t\leqslant -|\bx|\}$.
\end{itemize}
In \cite{Cuniberti} we have shown a procedure able to recover the ``jumps'' of the function $\cG$ across the
real cuts from the Fourier coefficients $g(\zeta_n;\cdot)$ of the function $\cG$ at imaginary time
regarded as initial data.

In \cite{Cuniberti} and in the present paper we replace the complex time variable $t+\rmi\gamma$ by 
$\tau=\rmi(t+\rmi\gamma)$,
and then putting $\tau=u+\rmi v$ ($u,v\in\R$), the initial data of the function
$\cG(\tau,\cdot)$ correspond to real values of $\tau$, and the real time $t$ corresponds
to $v$. Up to this change of notation, this analytic function $\cG(\tau,\cdot)$
plays the same role of the previously described two-point function of a boson or fermion field at
fixed $\bx$; the sole caution has to be taken concerning the ``jumps'' across the cuts,
noting that $u=-\gamma$. The only variable involved are $\tau$ and its conjugate
variable $\zeta$, the extra ``spectator variables'', denoted by the point $(\cdot)$,
may as well represent a fixed momentum (after Fourier transformation with respect to
the space variables) or the action on a test-function $f$.

In \cite{Cuniberti} we assumed as hypotheses the analytic structure of KMS condition, i.e., \\[+10pt]
{\bf Hypotheses:} The function $\cG(\tau,\cdot)$ ($\tau=u+\rmi v, u,v\in\R$) satisfies the
following properties:
\begin{itemize}\setlength\itemsep{0.5em}
\item[(a)] it is analytic in the open strips $k\beta<u<(k+1)\beta$
($v\in\R,k\in\Z,\beta=T^{-1}$) and continuous at the boundaries;
\item[(b)] it is periodic (antiperiodic) for bosons (fermions) with period $\beta$, i.e.,
\beq
\cG(\tau+\beta,\cdot)=\left\{
\begin{array}{ll}
\cG(\tau,\cdot) & ~ \mathrm{for~bosons,~} (\tau\in\C), \\[+6pt]
-\cG(\tau,\cdot) & ~ \mathrm{for~fermions,~} (\tau\in\C).
\end{array}
\right. \nonumber
\eeq
\item[(c)] \
$\sds\sup_{k\beta<u<(k+1)\beta}\left|\cG(u+\rmi v)\right|\leqslant C|v|^\alpha
\qquad (v\in\R;\, C,\alpha~\mathrm{constants}).
$
\end{itemize}

\begin{remark}
Strictly speaking condition (c) does not derive from the KMS conditions;
however, in \cite{Cuniberti} it plays a relevant role in order to derive the
``Froissart-Gribov'' equalities and, accordingly, the reconstruction of the
thermal Green functions at real times from those at imaginary times.
\end{remark}

In the present paper we invert the question, and pose the following problem.

\vspace*{0.2cm}

\begin{problem}
Given the trigonometric series
\beq
\cG(\tau,\cdot)=\frac{1}{\beta}
\sum_{n=-\infty}^{+\infty}g(\zeta_n;\cdot)\,e^{-\rmi\zeta_n\tau}
\quad \left(\zeta_n=n\,\frac{\pi}{\beta},\ n\ \mathrm{being\ even\ or\ odd\ integer}\right), \nonumber
\eeq
is it possible to find conditions on the coefficients $g(\zeta_n;\cdot)$ sufficient
to guarantee that $\cG(\tau,\cdot)$ satisfies conditions (a) and (b), i.e., the KMS analytic
structure?
\end{problem}


\newpage

\begin{thebibliography}{99}

\bibitem{Bernstein}
S. Bernstein,
D\'emonstration du th\'eor\`eme de Weierstrass fond\'ee sur le calcul des probabilit\'es,
\textsl{Comm. Kharkow Math. Soc.} \textbf{13} (1912), 1-2.

\bibitem{Bieberbach}
L. Bieberbach,
\textsl{Analytische Fortsetzung},
Springer, Berlin, 1955.

\bibitem{Boas}
R. P. Boas,
\textsl{Entire Functions},
Academic Press, New York, 1954.

\bibitem{Buchholz}
D. Buchholz and P. Junglas,
On the existence of equilibrium states in local quantum field theory,
\textsl{Comm. Math. Phys.} \textbf{121} (1989), 255-270. 

\bibitem{Cuniberti}
G. Cuniberti, E. De Micheli and G. A. Viano,
Reconstructing the thermal Green functions at real times from those at imaginary times,
\textsl{Comm. Math. Phys.} \textbf{216} (2001), 59-83.

\bibitem{DeMicheli1}
E. De Micheli and G. A. Viano,
Hausdorff moments, Hardy spaces, and power series,
\textsl{J. Math. Anal. Appl.} \textbf{234} (1999), 265-286.

\bibitem{DeMicheli2}
E. De Micheli and Viano,
On the solution of a class of Cauchy integral equations,
\textsl{J. Math. Anal. Appl.} \textbf{246} (2000), 520-543.

\bibitem{Bateman}
A. Erd\'elyi, W. Magnus, F. Oberhettinger and F. Tricomi (editors),
\textsl{Higher Trascendental Functions}, Volume II,
McGraw-Hill, New York, 1953.

\bibitem{Feller}
W. Feller,
\textsl{An Introduction to Probability Theory and its Applications}, Volume II,
John Wiley, New York, 1966.

\bibitem{Haag}
R. Haag, N. M. Hugenholtz and M. Winnink,
On the equilibrium states in quantum statistical mechanics,
\textsl{Comm. Math. Phys.} \textbf{5} (1967), 215-236.

\bibitem{Hadamard}
J. Hadamard,
\textsl{Lectures on the Cauchy Problem in Linear Differential Equations},
Yale University Press, New Haven, 1923.

\bibitem{Hoffman}
K. Hoffman,
\textsl{Banach Spaces of Analytic Functions},
Prentice-Hall, Englewood Cliff, NJ, 1962.

\bibitem{John}
F. John,
Continuous dependence on data for solutions of partial differential equations with a prescribed bound,
\textsl{Comm. Pure Appl. Math.} \textbf{13} (1960), 551-585.

\bibitem{Kahane}
J. P. Kahane,
A century of interplay between Taylor series, Fourier series and Brownian motion,
\textsl{Bull. Lond. Math. Soc.} \textbf{29} (1997), 257-279.

\bibitem{Kelley}
J. Kelley,
\textsl{General Topology},
Van Nostrand, Princeton, 1955.

\bibitem{LeRoy}
E. Le Roy,
Sur les series divergentes et les fonctions d\'efinies par un d\'eveloppement de Taylor,
\textsl{Ann. Fac. Sci. Toulouse Math.} \textbf{2} (1900), 317-430.

\bibitem{Lindelof}
E. Lindel\"of,
\textsl{Le Calcul des Residues et ses Applications \'a la The\'orie des Fonctions},
Chelsea, New York, 1947.

\bibitem{Lorentz}
G. G. Lorentz,
\textsl{Bernstein Polynomials},
University of Toronto Press, Toronto, 1953;
Second edition, Chelsea, New York, 1986.

\bibitem{Miller}
K. Miller \and G. A. Viano,
On the necessity of nearly-best-possible methods for analytic continuation
of scattering data,
\textsl{J. Math. Phys.} \textbf{14} (1973), 1037-1048.

\bibitem{Stein}
E. M. Stein and S. Wainger,
Analytic properties of expansions, and some variants of Parseval-Plancherel formulas,
\textsl{Ark. Mat.} \textbf{37} (1965), 553-567.

\bibitem{Szego}
G. Szeg\"o,
\textsl{Orthogonal Polynomials},
American Mathematical Society, Providence, 1959.

\bibitem{Tikhonov}
A. Tikhonov and V. Arsenine,
\textsl{M\'ethodes de R\'esolution de Probl\'emes Mal Pos\`es},
Mir, Moscow, 1976.

\bibitem{Watanabe}
T. Watanabe,
A probabilistic method in Hausdorff moment problem and Laplace-Stieltjes transform,
\textsl{J. Math. Soc. of Japan} \textbf{12} (1960), 192-206.

\bibitem{Widder}
D. V. Widder,
\textsl{The Laplace Transform},
Princeton University Press, Princeton, 1946.

\end{thebibliography}
\end{document}